\documentclass[acmtog]{acmart}

\AtBeginDocument{%
	}

\usepackage{makecell}
\usepackage{multirow}
\usepackage{colortbl}
\usepackage{enumitem}
\usepackage[ruled]{algorithm2e} 
\usepackage{xspace}

\newcommand{\ie}{\textit{i.e.,}\xspace}
\newcommand{\eg}{\textit{e.g.,}\xspace}
\newcommand{\minute}{\textit{min}}
\newcommand{\hour}{\textit{hrs}}

\newcommand{\desc}[1]{\parbox[c]{1\linewidth}{#1}}
\SetAlFnt{\small}
\SetAlCapFnt{\small}
\SetAlCapNameFnt{\small}
\SetAlCapHSkip{0pt}

\setcopyright{acmlicensed}
\copyrightyear{2025}
\acmYear{2025}
\acmDOI{10.1145/3757377.3764000}

\acmConference[SA Conference Papers '25]{SIGGRAPH Asia 2025 Conference Papers}{December 15--18, 2025}{Hong Kong, Hong Kong}
\acmBooktitle{SIGGRAPH Asia 2025 Conference Papers (SA Conference Papers '25), December 15--18, 2025, Hong Kong, Hong Kong}

\acmISBN{979-8-4007-2137-3/25/12}

\begin{document}
\title{Neural Hamiltonian Deformation Fields for Dynamic Scene Rendering}

\author{Hai-Long Qin}
\orcid{0009-0002-9114-9881}
\affiliation{%
	  \institution{Beijing University of Posts and Telecommunications}
	  \city{Beijing}
	  \country{China}}
\email{hailong.qin@bupt.edu.cn}

\author{Sixian Wang}
\orcid{0000-0002-0621-1285}
\affiliation{%
	\institution{Shanghai Jiao Tong University}
	\city{Shanghai}
	\country{China}}
\email{sxwang@sjtu.edu.cn}

\author{Guo Lu}
\orcid{0000-0001-6951-0090}
\affiliation{%
	\institution{Shanghai Jiao Tong University}
	\city{Shanghai}
	\country{China}}
\email{luguo2014@sjtu.edu.cn}

\author{Jincheng Dai}
\authornote{Corresponding author.}
\orcid{0000-0002-0310-568X}
\affiliation{%
	\institution{Beijing University of Posts and Telecommunications}
	\city{Beijing}
	\country{China}}
\email{daijincheng@bupt.edu.cn}

\begin{abstract}
Representing and rendering dynamic scenes with complex motions remains challenging in computer vision and graphics. Recent dynamic view synthesis methods achieve high-quality rendering but often produce physically implausible motions. We introduce NeHaD, a neural deformation field for dynamic Gaussian Splatting governed by Hamiltonian mechanics. Our key observation is that existing methods using MLPs to predict deformation fields introduce inevitable biases, resulting in unnatural dynamics. By incorporating physics priors, we achieve robust and realistic dynamic scene rendering. Hamiltonian mechanics provides an ideal framework for modeling Gaussian deformation fields due to their shared phase-space structure, where primitives evolve along energy-conserving trajectories. We employ Hamiltonian neural networks to implicitly learn underlying physical laws governing deformation. Meanwhile, we introduce Boltzmann equilibrium decomposition, an energy-aware mechanism that adaptively separates static and dynamic Gaussians based on their spatial-temporal energy states for flexible rendering. To handle real-world dissipation, we employ second-order symplectic integration and local rigidity regularization as physics-informed constraints for robust dynamics modeling. Additionally, we extend NeHaD to adaptive streaming through scale-aware mipmapping and progressive optimization. Extensive experiments demonstrate that NeHaD achieves physically plausible results with a rendering quality-efficiency trade-off. To our knowledge, this is the first exploration leveraging Hamiltonian mechanics for neural Gaussian deformation, enabling physically realistic dynamic scene rendering with streaming capabilities.
\end{abstract}

%
%
%


%
%


\begin{teaserfigure}
	\centering
	\includegraphics[width=1.0\linewidth]{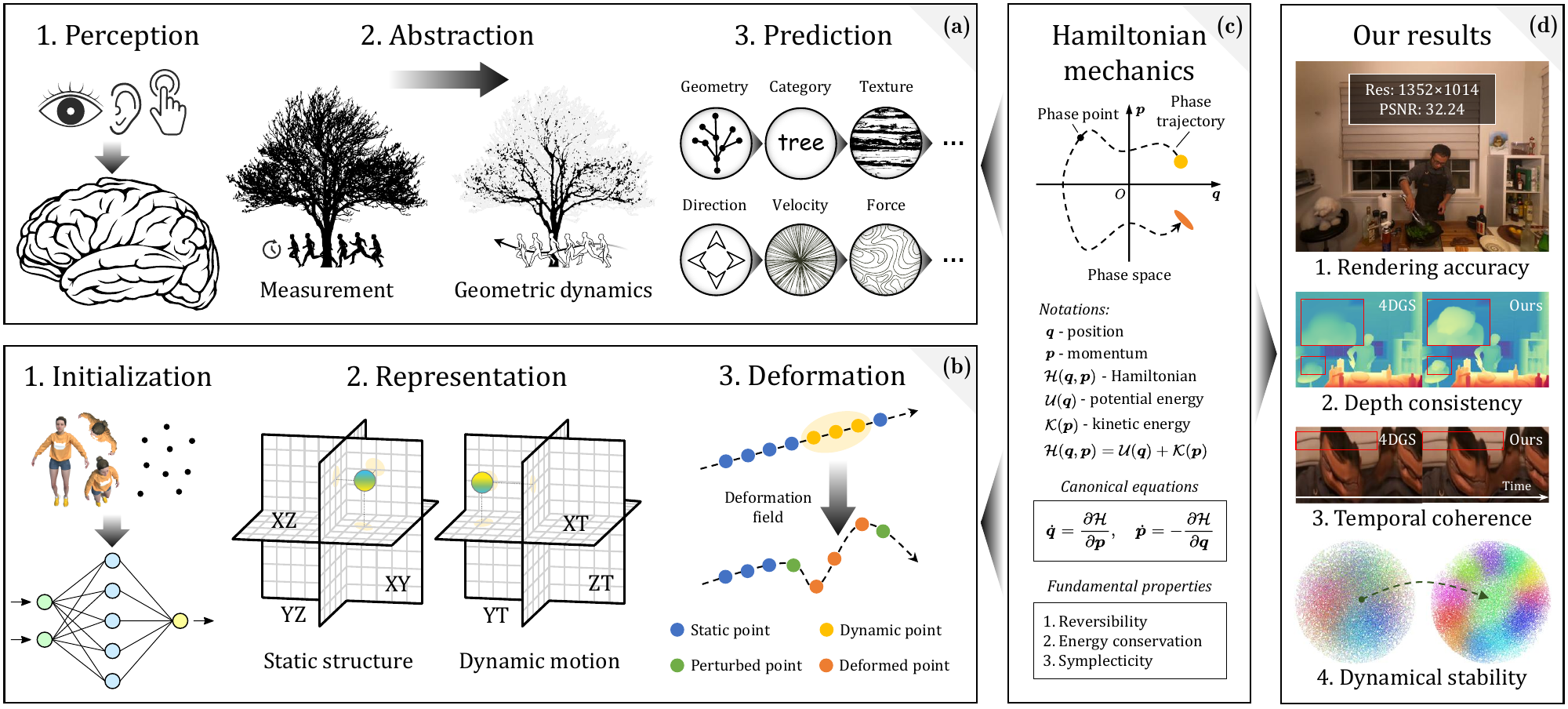}
	\caption{\textbf{Our method uses Hamiltonian mechanics to enhance dynamic Gaussian splatting for improved rendering quality and motion coherence.}
		(a) Human cognition process.
		(b) Scene rendering process.
		(c) Both processes follow physical laws, with Hamiltonian mechanics offering mathematical frameworks aligned with physical intuition.
		(d) Our method incorporates Hamiltonian mechanics as a physical prior in dynamic Gaussian splatting, improving rendering reality (4. different colors represent different deformations, and our method reduces overlap while keeping structural clusters).
	}
	\label{fig:teaser}
\end{teaserfigure}

\maketitle

\section{Introduction}
Representing and rendering dynamic scenes is pivotal for immersive imaging applications, pushing the boundaries of multimedia and graphics technologies including virtual reality (VR) and metaverse. Given discrete temporal video sequences, dynamic scene rendering aims to model scene dynamics and synthesize high-fidelity novel views at arbitrary timestamps in real-time. \emph{It confronts two primary challenges}: first, achieving high-fidelity reconstruction of complex dynamic scenes with rapid motions, and topological changes; second, maintaining real-time rendering efficiency with low training cost. However, existing methods struggle to simultaneously satisfy both ends, which is precisely the issue this article explores.

Recent advancements in dynamic scene rendering have been achieved mainly through methods based on Neural Radiance Field (NeRF)~\cite{nerf} and Gaussian Splatting~\cite{gaussian-splatting}. NeRF employs implicit neural fields to model static scenes and achieves photorealistic view synthesis. Its dynamic extensions either utilize deformation and canonical fields to model motions relative to canonical frames over time~\cite{dnerf, dynamic-nerf, hypernerf, nerf-ds, dynerf}, or store 4D volumes as explicit structural representations (\eg planes and hash encodings)~\cite{eg3d, instant-ngp, tensor4d, k-planes, hexplane, tineuvox, masked-spacetime-hashing, factorized-motion}. Despite progress in rendering quality, these methods suffer from slow rendering speeds due to their requirements for dense sampling along rays during rendering.

With the emergence of 3D Gaussian Splatting (3DGS)~\cite{gaussian-splatting}, high-fidelity and real-time rendering for static scenes becomes possible. Building on this milestone, several methods have extended 3DGS with the ability to model dynamic scenes~\cite{4dgs, deformable3d, spacetime-gaussians, beida4dgs, sc-gs, 3dgstream, grid4d, swift4d, saro-gs}. These approaches either use predefined functions for Gaussian deformation in scenes with sufficient viewpoints~\cite{gaussian-flow, spacetime-gaussians}, or employ neural networks to predict deformed Gaussian attributes~\cite{4dgs, deformable3d, sc-gs}. Nevertheless, they struggle to model temporally complex dynamics such as object appearances and disappearances, and fail to properly differentiate spatial and temporal deformations, leading to excessive coordinate overlap among deformed Gaussians and degraded rendering quality in complex motion scenarios.

To address the aforementioned challenges, we propose NeHaD, a neural deformation field for dynamic Gaussian Splatting governed by Hamiltonian mechanics. While Gaussian-based methods achieve real-time rendering, \emph{we argue that realistic rendering requires not only perceptual quality but also physically plausible dynamics}, which most current approaches have overlooked. Intriguingly as illustrated in Fig. \ref{fig:teaser}, the processes of human cognition and scene rendering both follow fundamental physical laws, particularly Hamiltonian mechanics for predicting system dynamics~\cite{noether1971invariant}, demonstrating the natural applicability of Hamiltonian principles to dynamic scene rendering. Moreover, Gaussian covariance matrices inherently exist on symplectic manifolds, making Hamiltonian mechanics a proper choice for mathematically reformulating Gaussian deformation fields.

Building on these insights, NeHaD enhances the Gaussian deformation fields of 4DGS~\cite{4dgs} through Hamiltonian mechanics. First, we replace the MLP-based deformation predictor with a Hamiltonian neural network (HNN)~\cite{hnn} to learn underlying conservation laws from data in an unsupervised manner. Through in-graph backpropagation of Hamiltonian gradients, HNN ensures stable and coherent deformations without sudden discontinuities while incurring minimal training overhead. Observing that most scene regions often remain static and requires no dynamic modeling~\cite{swift4d}, NeHaD introduces a soft masking mechanism that dynamically weights Gaussians based on their spatial-temporal energy states, enabling adaptive decomposition of static and dynamic components. This mechanism is driven by Boltzmann energy derived from primitives' deviation from equilibrium. To handle real-world dissipative forces such as friction, we employ second-order symplectic integration to preserve system symplecticity under perturbations. Additionally, we incorporate local rigidity constraints to avoid large rotations while preserving smaller, more natural rotations.

To extend NeHaD to bandwidth-constrained streaming applications, we incorporate scale-aware anisotropic mipmapping for anti-aliasing and layered progressive optimization for level-of-detail (LOD) rendering. We extensively evaluate our approach on monocular and multi-view dynamic scene datasets containing both synthetic and real-world scenes. Both quantitative and qualitative results demonstrate that our method achieves physically plausible rendering with an improved quality-efficiency trade-off, effectively modeling complex system dynamics across diverse scenes. Our contributions are summarized below:
\begin{itemize}
	\item We propose a Hamiltonian-based neural deformation field for dynamic Gaussian Splatting. Using a single HNN with attribute-specific adapters, we implicitly learn conservation laws from data for physically plausible deformations.
	
	\item We introduce Boltzmann equilibrium decomposition to adaptively separate static and dynamic Gaussians. Meanwhile, we enhance deformation modeling through physics-informed constraints, ensuring robust and realistic rendering.
	
	\item We extend NeHaD to streaming with scale-aware mipmapping and progressive optimization. Extensive experiments demonstrate our improvements in rendering reality. To our knowledge, this is the first exploration leveraging Hamiltonian mechanics for neural Gaussian deformation.
\end{itemize}

\section{Related Work}

\subsection{NeRF-based Dynamic Scene Rendering}
NeRF~\cite{nerf} reconstructs light fields of static scenes through implicit neural representations, achieving significant visual improvements. To extend NeRF to dynamic scenes, implicit deformation fields are applied to static models~\cite{dnerf}.

Various approaches have been developed to model dynamic scenes more accurately. Some methods segment scenes into components with different temporal behaviors~\cite{dynamic-nerf, nr-nerf}, while others incorporate higher-dimensional latent codes~\cite{dynerf, hypernerf, nerfies} with additional supervision techniques including optical flow across frames~\cite{fsdnerf, forwardflowdnerf, saff, li2021neural} and motion mask constraints~\cite{nerf-ds}. Meanwhile, special attention to rigid objects is particularly important due to their prevalence and unique physical properties~\cite{nr-nerf, star}.

Recent research has addressed challenging scenarios including dynamic human modeling~\cite{ndr}, specular objects~\cite{nerf-ds}, streaming~\cite{nerfplayer}, and scenes without known camera poses~\cite{rodynrf}. However, implicit MLP-based representations suffer from over-smoothing and require computationally expensive training. Explicit representations, such as Triplanes~\cite{eg3d} and Hash Encoding~\cite{instant-ngp}, address these limitations by improving both visual quality and training efficiency. A popular approach for dynamic scene rendering decomposes 4D inputs into six 2D planes~\cite{k-planes, hexplane, tensor4d, 4k4d, factorized-motion}.

\subsection{Gaussian-based Dynamic Scene Rendering}
3DGS~\cite{gaussian-splatting} represents static scenes using Gaussian primitives, achieving fast training and high visual quality. For dynamic scenes, two main approaches have emerged: using 4D Gaussians or deforming Gaussians with predefined functions~\cite{fudan4dgs, beida4dgs, gaussian-flow, dynamic-3d-gaussians}, and deforming 3D Gaussian attributes through neural networks~\cite{deformable3d, 4dgs, dynmf, md-splatting, gaufre, gags, cogs, 3dgstream}.

While fully MLP-based Gaussian deformation fields achieve high quality~\cite{deformable3d}, they suffer from over-smoothing, leading to poor detail rendering in complex scenes. Explicit methods like 4DGS~\cite{4dgs} employ plane-based deformation fields, but their low-rank assumptions cause feature overlap and rendering artifacts. Recent advances aim to address these limitations: Motion-aware methods~\cite{zhu2024motiongs, guo2024motion} leverage optical flow constraints to guide Gaussian deformations for robust rendering; SaRO-GS~\cite{saro-gs} uses scale-aware residual fields with explicit-implicit blending for better spatial-temporal correlations; Grid4D~\cite{grid4d} decomposes 4D encoding into spatial and temporal 3D hash encodings without low-rank assumptions; and Swift4D~\cite{swift4d} separates Gaussians into static and dynamic components, applying deformation only to dynamic points.

Building on this static-dynamic decomposition paradigm, our method distinguishes itself through Boltzmann energy-aware soft decomposition of Gaussians and physics-informed deformations guided by Hamiltonian mechanics.

\section{Preliminaries}

\subsection{Gaussian Splatting}
\noindent \textbf{Static Gaussian Splatting.} 3DGS~\cite{gaussian-splatting} has emerged as a powerful static scene representation method, known for its high training speed and visual quality. Given input images with corresponding camera parameters, 3DGS explicitly represents scene geometry and appearance using anisotropic ellipsoids (\ie Gaussian primitives), contrasting with NeRF's implicit neural representation strategy. Each Gaussian primitive $\mathcal{G}$ is parameterized by position $\boldsymbol{\mu} \in \mathbb{R}^3$, covariance matrix $\boldsymbol{\Sigma} \in \mathbb{R}^{3\times3}$, color $\boldsymbol{c} \in \mathbb{R}^n$, and opacity $\alpha \in \mathbb{R}$. The covariance matrix is factorized into scaling vector $\boldsymbol{s} \in \mathbb{R}^3$ and rotation quaternion $\boldsymbol{r} \in \mathbb{R}^4$, while color is represented by spherical harmonic (SH) coefficients with $n$ SH functions. For rendering, 3DGS employs the tile-based differentiable rasterization.

\noindent \textbf{Dynamic Gaussian Splatting.} 3DGS can be extended to 4D dynamic scenes by incorporating the temporal dimension. Instead of applying 3DGS to individual frames, 4DGS~\cite{4dgs} uses plane-based deformation fields for real-time dynamic rendering.

Given the camera view matrix $\boldsymbol{V}$, a novel-view image $\boldsymbol{X}$ is rendered as $\boldsymbol{X} = \mathcal{R}(\boldsymbol{V}, \mathcal{G}')$, where $\mathcal{G}' = \mathcal{G} + \Delta \mathcal{G}$ represents deformed Gaussians and $\mathcal{R}$ denotes differentiable rasterization.

The deformation $\Delta \mathcal{G}$ is predicted by deformation field at timestamp $t$. Specifically, a hex-plane encoder $\mathcal{E}$ extracts spatial-temporal features, which are then processed by MLP decoder $\mathcal{D}$ to predict deformations, \ie $\Delta \mathcal{G} = \mathcal{D}(\mathcal{E}(\mathcal{G}, t))$.

The hex-plane factorization employs six planes: spatial-only planes $\boldsymbol{P}_{XY}$, $\boldsymbol{P}_{XZ}$, $\boldsymbol{P}_{YZ}$ and spatial-temporal planes $\boldsymbol{P}_{XT}$, $\boldsymbol{P}_{YT}$, $\boldsymbol{P}_{ZT}$. For a 4D coordinate $\boldsymbol{u} = (x, y, z, t)$, features are obtained by:
\begin{equation}
	\boldsymbol{f}(\boldsymbol{u}) = \prod_{k \in K} \boldsymbol{f}(\boldsymbol{u})_k = \prod_{k \in K} \psi(\boldsymbol{P}_k, \pi_k(\boldsymbol{u})),  \label{eq:k-planes}
\end{equation}
where $\pi_k$ projects $\boldsymbol{u}$ onto plane $k$, and $\psi$ means bilinear interpolation.

Using extracted features $\boldsymbol{f}$, a multi-head decoder $\mathcal{D}=\{\mathcal{D}_{\boldsymbol{\mu}}, \mathcal{D}_{\boldsymbol{s}}, \mathcal{D}_{\boldsymbol{r}}\}$ predicts Gaussian deformations for position ($\Delta \boldsymbol{\mu} = \mathcal{D}_{\boldsymbol{\mu}}(\boldsymbol{f})$), scaling ($\Delta \boldsymbol{s} = \mathcal{D}_{\boldsymbol{s}}(\boldsymbol{f})$), and rotation ($\Delta \boldsymbol{r} = \mathcal{D}_{\boldsymbol{r}}(\boldsymbol{f})$). The final deformed Gaussian is $\mathcal{G}' = \{\boldsymbol{\mu}', \boldsymbol{s}', \boldsymbol{r}', \alpha, \boldsymbol{c}\}$ where $\boldsymbol{\mu}' = \boldsymbol{\mu} + \Delta \boldsymbol{\mu}$, $\boldsymbol{s}' = \boldsymbol{s} + \Delta \boldsymbol{s}$, and $\boldsymbol{r}' = \boldsymbol{r} + \Delta \boldsymbol{r}$.

\subsection{Hamiltonian Mechanics}  \label{sec:hamilton}
Given position-momentum coordinates $(\boldsymbol{q},\boldsymbol{p})$, $\boldsymbol{q}\in \mathbb{R}^d$, $\boldsymbol{p}\in\mathbb{R}^d$, where $d$ is the degrees of freedom, the Hamiltonian $\mathcal{H}(\boldsymbol{q},\boldsymbol{p})$ represents total system energy:
\begin{equation}
	\mathcal{H}(\boldsymbol{q},\boldsymbol{p})=\mathcal{U}(\boldsymbol{q})+\mathcal{K}(\boldsymbol{p}), \label{eq:huk}
\end{equation}
where $\mathcal{U}(\boldsymbol{q})$ and $\mathcal{K}(\boldsymbol{p})$ are potential and kinetic energy, respectively.

The Hamiltonian canonical equations is thus defined following
\begin{equation}
	\dot{\boldsymbol{q}}=\frac{\mathrm{d} \boldsymbol{q}}{\mathrm{d} t}=\frac{\partial \mathcal{H}}{\partial \boldsymbol{p}}, \quad \dot{\boldsymbol{p}}=\frac{\mathrm{d} \boldsymbol{p}}{\mathrm{d} t}=-\frac{\partial \mathcal{H}}{\partial \boldsymbol{q}}.  \label{eq:hamilton}
\end{equation}

These equations ensure three fundamental properties:

\emph{Reversibility:} The time-evolution mapping is invertible.

\emph{Energy Conservation:} Total energy remains invariant.

\emph{Symplecticity:} Hamiltonian mechanics preserves volume in $(\boldsymbol{q}, \boldsymbol{p})$ space (\ie Liouville's Theorem). The symplectic gradient $\boldsymbol{S}_{\mathcal{H}}(\boldsymbol{q}, \boldsymbol{p})=(\frac{\partial \mathcal{H}}{\partial \boldsymbol{p}}, -\frac{\partial \mathcal{H}}{\partial \boldsymbol{q}})$ keeps constant energy while evolving the system.

Hamiltonian neural networks (HNNs) learn a parametric function for $\mathcal{H}$ from data instead of computing $\boldsymbol{S}_{\mathcal{H}}$ analytically. During forward pass, the network outputs a scalar energy value. Then the $L_2$ loss enforces Hamiltonian constraints through in-graph gradients:
\begin{equation}
	\mathcal{L}=\left\|\frac{\partial \mathcal{H}_{\boldsymbol{\theta}}}{\partial \boldsymbol{p}}-\frac{\partial \boldsymbol{q}}{\partial t}\right\|_{2}+\left\|\frac{\partial \mathcal{H}_{\boldsymbol{\theta}}}{\partial \boldsymbol{q}}+\frac{\partial \boldsymbol{p}}{\partial t}\right\|_{2}.
\end{equation}

Note that HNNs preserve a quantity close to, but not exactly equivalent to, total energy. This limits their applicability to systems with non-conservative forces, such as friction.

\section{Methodology}

\begin{figure*}[t]
	\centering
	\includegraphics[width=1\textwidth]{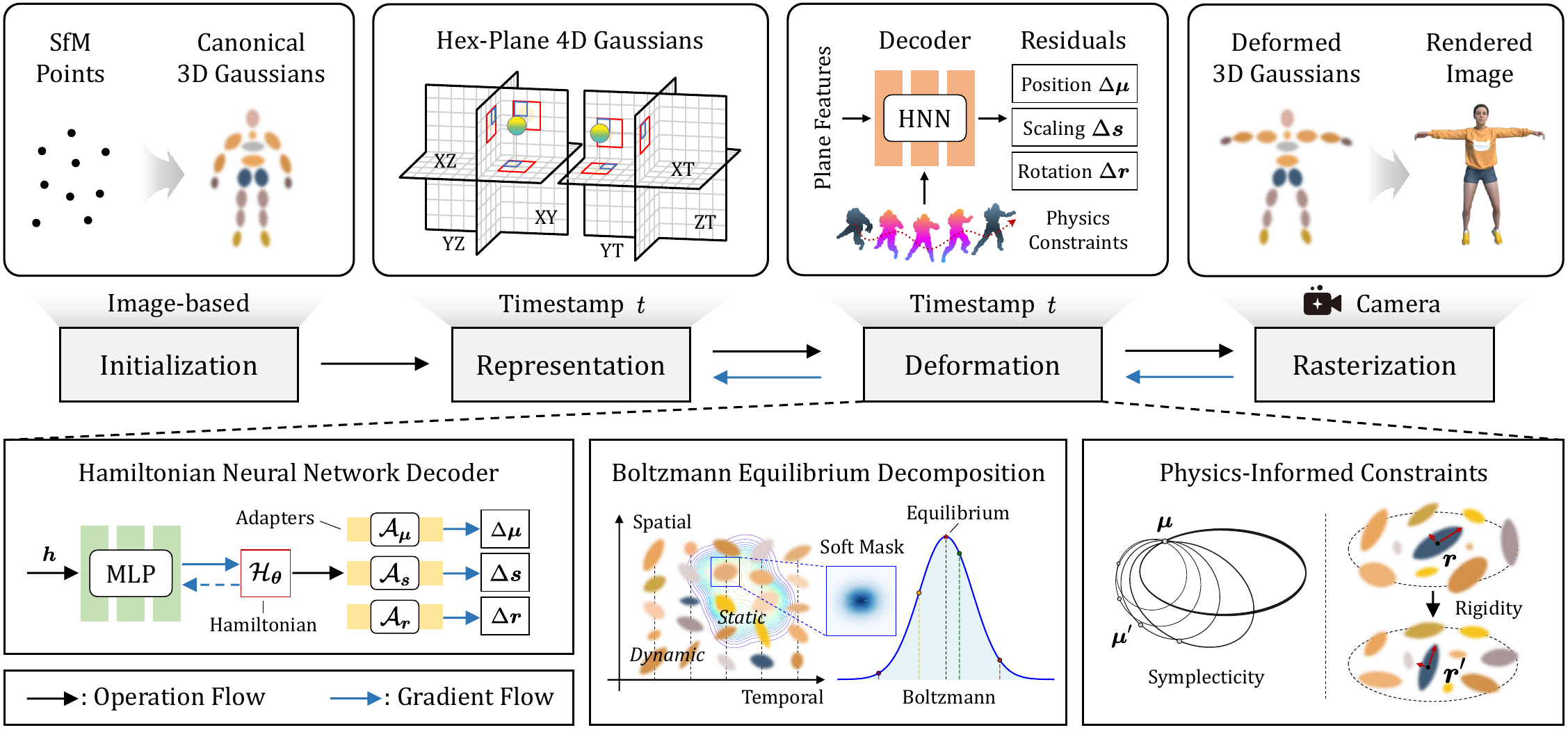}
	\caption{\textbf{Overall pipeline of NeHaD.} (from left to right) An HNN with MLP baseline learns conservation laws from data. Through backpropagation of Hamiltonian gradients, the HNN optimizes vector fields and predicts Gaussian deformations (position, scaling, rotation) via adapters. The Boltzmann equilibrium decomposition decides which primitives should not be deformed with soft masks, \ie smaller deviations from equilibrium maintain static during deformation. Physics-informed constraints including symplectic integration and rigidity regularization are used to preserve system properties.}
	\label{fig:pipeline}
\end{figure*}

\subsection{Neural Hamiltonian Deformation Fields}  \label{sec:hnn}
Existing dynamic Gaussian splatting methods rely on purely data-driven MLP decoders to predict deformation fields, which often leads to physically implausible motions such as abrupt appearance/disappearance of primitives, unrealistic trajectory discontinuities, and violation of energy conservation principles. These limitations arise from the lack of physics-informed inductive biases in conventional neural architectures.

To address this fundamental issue, we propose to replace the standard MLP-based deformation decoder with a HNN that inherently respects the underlying physical laws governing dynamic systems. The key insight driving our approach is that Gaussian primitives naturally exist in a phase space where their positions and momenta evolve according to Hamiltonian dynamics. The overall pipeline of our method is illustrated in Fig. \ref{fig:pipeline}.

Rather than completely abandoning MLPs, our HNN framework leverages MLPs as differentiable baselines while imposing physics-informed structural constraints. HNN optimizes symplectic gradients $\boldsymbol{S}_{\mathcal{H}}$ (mentioned in Sec. \ref{sec:hamilton}) instead of ordinary gradients. The key distinction lies in gradient computation: while traditional MLPs compute gradients only during backpropagation for parameter updates, HNN requires additional forward-mode gradient computation to construct vector fields from scalar potentials.

Given a Gaussian primitive with position $\boldsymbol{\mu}_i$ at timestamp $t$, we extract spatial-temporal features $\boldsymbol{f}_i$ using a hex-plane encoder $\mathcal{E}$ following Eq. \eqref{eq:k-planes}. These features are processed via a differentiable MLP $\mathcal{M}$ composed of linear layers and ReLU activations, with depth $D$ and width $W$, mapping features to latent representations $\boldsymbol{h}_i = \mathcal{M}(\boldsymbol{f}_i) \in \mathbb{R}^W$ that serve as the input to the Hamiltonian dynamics.

Note that the $(\boldsymbol{q}, \boldsymbol{p})$ notation follows Hamiltonian conventions but is not the actual input. Explicitly defining position-momentum coordinates for high-dimensional Gaussian primitives in neural rendering is intractable. Therefore, we employ implicit latent representations as the input to the HNN instead\footnote{The original HNN paper demonstrated the concept on simple systems where the "correct" coordinates (like the angle and angular velocity of a pendulum) are already known. In a complex, high-dimensional scene represented by thousands of Gaussians, defining an explicit $(\boldsymbol{q}, \boldsymbol{p})$ for the entire system is intractable, which is the direct reason why NeHaD uses the implicit features as the input. In fact, our NeHaD extends the standard HNN with the concept of a \emph{learned latent phase space}, where the latent representations $\boldsymbol{h}_i$ act as the generalized coordinates for the dynamical system.}. The hex-plane features $\boldsymbol{f}_i$ encode Gaussian spatial-temporal states containing necessary geometric and kinetic information for phase space construction, hence the latent representations $\boldsymbol{h}_i = \mathcal{M}(\boldsymbol{f}_i)$ can construct generalized coordinates with implicit position-momentum coupling, maintaining information completeness while avoiding $\frac{\mathrm{d} \boldsymbol{q}}{\mathrm{d} t}$ calculation, thereby reducing complexity.

Simulating Hamiltonian dynamics requires modeling a vector field $\boldsymbol{v}$, which can be decomposed into a conservative field $\boldsymbol{v}_{c}$ and a solenoidal field $\boldsymbol{v}_{s}$, \eg $\boldsymbol{v} = \boldsymbol{v}_{c} + \boldsymbol{v}_{s}$, as proven in the Appendix~\ref{sec:proof}. The objective of HNN is to learn these two fields. However, directly learning decomposed vector fields $\boldsymbol{v}_{c}$ and $\boldsymbol{v}_{s}$ through HNN is challenging due to two main issues: 1) the high dimensionality can lead to mode collapse, and 2) there is a lack of energy conservation constraints, as energy is a scalar, making it difficult to ensure that Eq. \eqref{eq:huk} holds during vector field learning. To address these challenges, we instead learn two scalar functions $F_1$ and $F_2$. These functions generate vector fields via automatic differentiation guaranteeing the desired physical properties: $\boldsymbol{v}_{c}$ is inherently conservative ($\nabla \times \boldsymbol{v}_c = \boldsymbol{0}$), preserving energy and generating gradient-based motions, while $\boldsymbol{v}_{s}$ is solenoidal ($\nabla \cdot \boldsymbol{v}_s = 0$), preserving volume and generating rotational motions.

Specifically, the HNN decoder learns two scalar functions $F_1(\boldsymbol{h}_i)$ and $F_2(\boldsymbol{h}_i)$ that generate $\boldsymbol{v}_{c}$ and $\boldsymbol{v}_{s}$ as follows:
\begin{align}
	\boldsymbol{v}_{c} = \nabla_{\boldsymbol{h}_i} &F_1(\boldsymbol{h}_i) \boldsymbol{I}, \quad \boldsymbol{v}_{s} = \nabla_{\boldsymbol{h}_i} F_2(\boldsymbol{h}_i) \boldsymbol{M}^{\top},\label{eq:field}\\
	&\boldsymbol{M} = \begin{bmatrix} \boldsymbol{0}_{d \times d} & \boldsymbol{I}_{d \times d} \\ -\boldsymbol{I}_{d \times d} & \boldsymbol{0}_{d \times d} \end{bmatrix},
\end{align}
where $\boldsymbol{M}$ is the permutation tensor keeping the symplectic structure and $\boldsymbol{I}$ is the identity matrix. We use lightweight attribute-specific adapters $\mathcal{A}_{\boldsymbol{\mu}}$, $\mathcal{A}_{\boldsymbol{s}}$, and $\mathcal{A}_{\boldsymbol{r}}$ (implemented as linear layers) to process HNN-generated vector fields, maintaining standard dimensionality of each Gaussian attribute. Finally, the predicted deformations of each Gaussian attributes are $\Delta \boldsymbol{\mu}_i = \mathcal{A}_{\boldsymbol{\mu}} \boldsymbol{v}$, $\Delta \boldsymbol{s}_i = \mathcal{A}_{\boldsymbol{s}} \boldsymbol{v}$, $\Delta \boldsymbol{r}_i = \mathcal{A}_{\boldsymbol{r}} \boldsymbol{v}$, respectively.

\subsection{Boltzmann Equilibrium Decomposition}  \label{sec:bed}
While HNN ensures physically consistent deformation fields, a critical challenge remains: determining which Gaussian primitives require dynamic modeling versus those that remain static throughout the temporal sequence. Indiscriminate application of deformation to all primitives leads to computational inefficiency and potential artifacts in stable regions. To address this challenge, we introduce Boltzmann Equilibrium Decomposition (BED).

The insight driving BED is that only Gaussian primitives away from their equilibrium states should be activated for deformation. We formalize this through statistical mechanics, constructing soft masks based on Boltzmann energy distributions that adaptively filter out primitives useless to deformation. Our decomposition strategy employs two complementary mechanisms tailored to the distinct visual characteristics of different Gaussian attributes.

\noindent \textbf{Spatial-Temporal Decomposition for Position Dynamics.} Position deformation requires careful spatial-temporal filtering because not all Gaussians need to deform at every moment, and not all spatial locations are equally important at each timestamp. The decomposition pattern exhibits dual selectivity:

\emph{Spatial selectivity.} At any given timestamp, only a subset of Gaussians significantly contribute to the deformation, while others should remain static to maintain scene stability.

\emph{Temporal selectivity.} For a given Gaussian primitive, its activation varies across different timestamps, \ie it may be dormant at certain frames but reactivated when motion patterns require.

We model this dual selectivity by constructing a phase space where each Gaussian primitive's equilibrium state is determined by both its spatial position and temporal context. For the $i$-th Gaussian at position $\boldsymbol{\mu}_i$ and timestamp $t$, the equilibrium deviation is:
\begin{equation}
	\Delta d_i = \frac{\|\boldsymbol{\mu}_i - \boldsymbol{\mu}_{eq}^{(i)}\|_2}{\sigma_s}, \quad \Delta \tau_i = \frac{t - t_{eq}^{(i)}}{\sigma_t},
\end{equation}
where $\boldsymbol{\mu}_{eq}^{(i)}$ and $t_{eq}^{(i)}$ represent the spatial and temporal equilibrium states for the $i$-th Gaussian primitive, respectively, and $\sigma_t$, $\sigma_s$ control the sensitivity scales. The spatial and temporal equilibrium states are both initialized as Gaussian attributes and optimized during training. Distance from equilibrium determines dynamic/static status - larger distances indicate greater deviation from equilibrium, meaning more intense motion, hence more likely to be dynamic primitives. The combined spatial-temporal energy deviation follows a harmonic oscillator model for simplicity:
\begin{equation}
	E_{st}^{(i)} = \frac{1}{2}(\Delta d_i^2 + \Delta \tau_i^2) + \lambda \Delta d_i \Delta \tau_i,
\end{equation}
where the coupling term $\lambda \Delta d_i \Delta \tau_i$ captures the interdependence between spatial and temporal deviations, reflecting that spatial motion patterns are inherently coupled with temporal evolution. The position equilibrium mask follows the Boltzmann distribution:
\begin{equation}
	M_{pos}^{(i)} = (1 - \gamma) \cdot \exp(-\beta E_{st}^{(i)}) + \gamma,
\end{equation}
where $\beta=1/T$ is the inverse temperature that controls the sharpness of the energy distribution, and $\gamma$ ensures minimum dynamic responsiveness to prevent complete deactivation of primitives.

This mask naturally achieves the desired dual selectivity: primitives near equilibrium keep static, while those far from equilibrium are subjected to stronger deformations.

\noindent \textbf{Temporal-Only Decomposition for Scaling Dynamics.} Scaling dynamics exhibit fundamentally different visual characteristics compared to position dynamics. The visual granularity of scaling changes is much smaller than positional changes, primarily affecting surface smoothness and appearance detail rather than global scene structure. Consequently, our decomposition strategy for scaling follows a different principle:

\emph{Spatial universality.} At any given timestamp, nearly all Gaussian primitives should participate in scaling deformation, as scaling contributes to the fine-grained surface details across the entire scene.

\emph{Temporal selectivity.} Different timestamps have varying scaling requirements, \ie some moments benefit from active scaling to enhance surface quality, while others require minimal scaling to preserve texture authenticity.

Based on this analysis, we apply temporal-only decomposition for scaling dynamics, eliminating spatial constraints while maintaining temporal equilibrium-based separation.

For scaling, the temporal energy deviation is modeled as:
\begin{equation}
	E_{t}^{(i)} = \frac{1}{2}\left(\frac{t - t_{eq}^{(i)}}{\sigma_{t}}\right)^2,
\end{equation}
where $t_{eq}^{(i)}$ is initialized and optimized dynamically as an Gaussian attribute, and $\sigma_{t}$ controls the temporal sensitivity around the equilibrium time $t_{eq}^{(i)}$. Then the scaling equilibrium mask becomes:
\begin{equation}
	M_{scale}^{(i)} = (1 - \gamma) \cdot \exp(-\beta E_{t}^{(i)}) + \gamma.
\end{equation}

This temporal-only approach ensures that scaling deformation adapts to the natural temporal rhythm of the scene while maintaining spatial universality for surface detail preservation.

Finally, the BED mechanism integrates into the deformation pipeline through equilibrium-aware blending, applying different decomposition strategies based on attribute characteristics:
\begin{align}
	\boldsymbol{\mu}_i' &= \boldsymbol{\mu}_i + \Delta\boldsymbol{\mu}_i \odot (1 - M_{pos}^{(i)}), \label{eq:position}\\
	\boldsymbol{s}_i' &= \boldsymbol{s}_i + \Delta\boldsymbol{s}_i \odot (1 - M_{scale}^{(i)}),
\end{align}
where $\Delta\boldsymbol{\mu}_i$ and $\Delta\boldsymbol{s}_i$ are the HNN-generated deformation predictions, and $\odot$ is the Hadamard  product. These deformation residuals can be further regularized and refined using specialized physics-informed constraints, which are introduced in the following subsection.

\subsection{Physics-Informed Constraints}  \label{sec:pic}
While the HNN decoder and BED mechanism provide physically consistent deformation predictions and energy-based primitive selection, the temporal evolution of Gaussian attributes requires additional physics-informed constraints to ensure stable and directed dynamics. We introduce two specialized constraints guiding the deformation process: Second-order Symplectic Integration for position dynamics and Local Rigidity Regularization for rotation dynamics.

\noindent \textbf{Second-order Symplectic Integration.} The temporal integration of position dynamics presents a fundamental challenge in physics-based simulation~\cite{sholokhov2023physics}: maintaining energy conservation and system stability over extended time sequences. Standard Euler integration and higher-order Runge-Kutta methods~\cite{ode-solver}, commonly used in neural operators, suffer from energy drift and numerical instability, particularly for systems governed by Hamiltonian mechanics, due to violating symplecticity.

To ensure the long-term stability of our physics-based simulation, we use a symplectic integrator to update the position of each Gaussian. We have adopted the \emph{Position Verlet} integration scheme, which requires both the current velocity and acceleration of a primitive. Specifically, our HNN decoder is structured to provide both of these physical quantities:

\emph{Velocity.} The primary deformation vector, $\Delta \boldsymbol{\mu}_i$, is interpreted as the instantaneous velocity of the $i$-th primitive at a given timestamp.

\emph{Force.} The HNN also learns a latent potential energy landscape as part of its Hamiltonian. The force acting on the $i$-th primitive is the negative gradient of this potential ($-\nabla_{\boldsymbol{q}}\mathcal{U}(\boldsymbol{q}_i)$). This force guides the Gaussian position updates within the symplectic integration, treating the primitives as particles driven toward energy minimization. This regularizes the system and helps enforce conservation laws, even in the presence of dissipative effects. We obtain this force vector $\boldsymbol{F}_i$ directly from the conservative component ($\boldsymbol{v}_c$) of the HNN's output vector field (see Eq. \eqref{eq:field}), which induces a change in the primitive's momentum (\ie acceleration) and is calculated via automatic differentiation.

Assuming unit mass ($m=1$) for each Gaussian primitive\footnote{Newton's second law of motion is expressed mathematically as $\boldsymbol{F} = m \boldsymbol{a}$, where $\boldsymbol{F}$ represents the net force acting on an object, $m$ denotes the object's mass, and $\boldsymbol{a}$ represents the resulting acceleration.}, the force is equivalent to acceleration ($\boldsymbol{a}_i = \boldsymbol{F}_i$). We can then substitute these terms into the standard Position Verlet integration formula:

\begin{equation}
	\tilde{\boldsymbol{\mu}}_{i} = \boldsymbol{\mu}_i + \Delta t \cdot \Delta \boldsymbol{\mu}_i + \frac{(\Delta t)^2}{2} \boldsymbol{F}_i.  \label{eq:leapfrog}
\end{equation}

This formulation bridges our data-driven predictions with a principled integration scheme, ensuring that the position updates of the Gaussians respects the underlying physical laws (\eg symplectic structures) learned by the HNN. Note that Eq. \eqref{eq:leapfrog} should be applied before Eq. \eqref{eq:position}, \ie $\boldsymbol{\mu}_i^{\prime} = \tilde{\boldsymbol{\mu}}_{i} \odot (1 - M_{pos}^{(i)}) + \boldsymbol{\mu}_i \odot M_{pos}^{(i)}$. A detailed description of the algorithm pipeline can be found in the Appendix~\ref{sec:alg}.

\noindent \textbf{Local Rigidity Regularization.} Rotation dynamics in Gaussian Splatting present unique challenges due to the quaternion parameterization and the need to maintain local geometric coherence.

To address these challenges, we implement the local rigidity regularization motivated by As-Rigid-As-Possible (ARAP)~\cite{arap}, which constrains rotation updates to preserve local rigidity while allowing global flexibility. The ARAP principle ensures that local neighborhoods of Gaussian primitives undergo near-rigid transformations, preventing unnatural deformations while maintaining overall scene dynamics. To ensure temporally smooth and geometrically stable rotations, we regularize the rotation updates to prevent unnaturally large changes between time steps. This is achieved by clamping the magnitude of the rotation angle.

For the $i$-th Gaussian primitive with quaternion rotation $\boldsymbol{r}_i=[w_i, x_i, y_i, z_i]^{\top}$, we apply simplified ARAP-constrained rotation updates that do not require minimizing energy functions as in~\cite{arap, sc-gs}. The rotation head of our HNN decoder predicts a unit quaternion, $\Delta \boldsymbol{r}_i$, which represents the change in orientation from time $t$ to $t + \Delta t$. We can decompose this quaternion increment into its scalar and vector components:
\begin{equation}
	\Delta \boldsymbol{r}_i = [\Delta w_i, \Delta \boldsymbol{g}_i]^{\top}, \quad \Delta \boldsymbol{g}_i = [\Delta x_i, \Delta y_i, \Delta z_i]^{\top},
\end{equation}
where $\Delta w_i$ is the scalar part, indicating the rotation magnitude, and $\Delta \boldsymbol{g}_i$ is the vector part, representing the rotation axis scaled by angle. We convert this quaternion into its axis-angle representation to get the total rotation angle $\phi_i = 2 \cdot \mathrm{atan2} (\| \Delta \boldsymbol{g}_i \|, \Delta w_i)$ in practice\footnote{Using ``$\mathrm{atan2}$'' is more numerically stable than ``$\arctan$'', especially for angles greater than 90 degrees, and is the standard practice.}.

To prevent abrupt rotational changes, we apply a smooth limiting function that constrains large rotations while preserving small natural rotations to get a clamped angle, \ie $\phi_i^{\prime} = \phi_{max} \cdot \tanh\left(\frac{\phi_i}{\phi_{max}}\right)$, where $\phi_{max}$ defines the maximum allowable rotation per timestamp, and $\tanh(\cdot)$ provides smooth saturation effect. Further on these steps, the constrained quaternion increment $\Delta \boldsymbol{r}_i^{\prime}$ is constructed as:
\begin{equation}
	\Delta \boldsymbol{r}_i^{\prime} = \left[\cos\left(\frac{\phi_i^{\prime}}{2}\right), \sin\left(\frac{\phi_i^{\prime}}{2}\right) \frac{\Delta \boldsymbol{g}_i}{\|\Delta \boldsymbol{g}_i\|}\right]^{\top},
\end{equation}
and the final rotation is thus obtained through quaternion multiplication followed by normalization $\boldsymbol{r}_i^{\prime} = \mathcal{N}(\boldsymbol{r}_i \otimes \Delta \boldsymbol{r}_i^{\prime})$, where $\otimes$ and $\mathcal{N}(\cdot)$ denote quaternion multiplication and normalization.

\section{Experiments}

\begin{figure}[t]
	\centering
	\includegraphics[width=0.98\columnwidth]{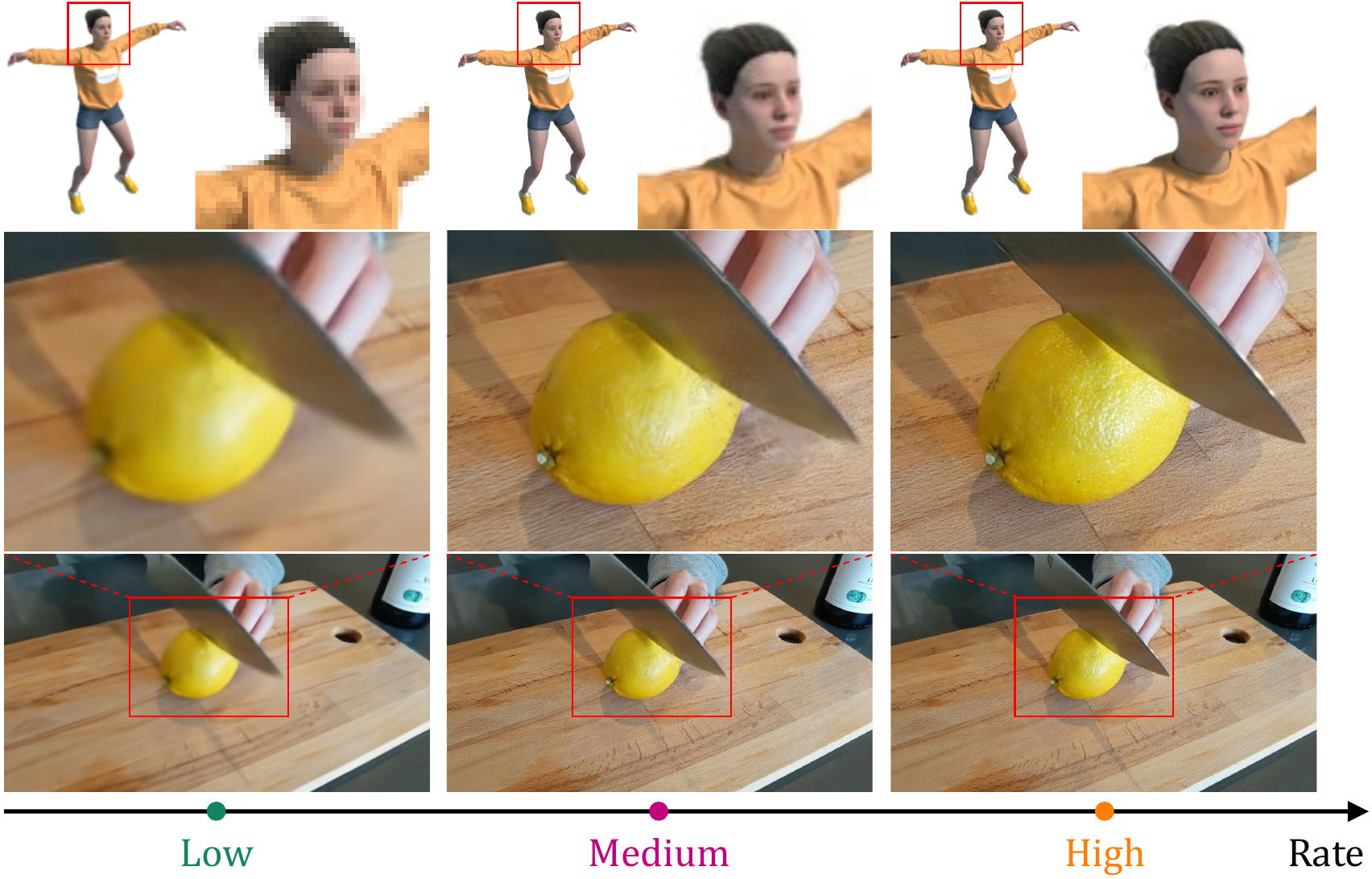}
	\caption{\textbf{Qualitative results of adaptive streaming.} NeHaD adapts effectively to progressive streaming across varying allocated rates.}
	\label{fig:streaming-results}
\end{figure}

\begin{table}[t]
	\caption{\textbf{Quantitative results.} Averaged metrics over all scenes in the respective datasets. The best and second-best results are indicated in \textbf{bold} and \underline{underlined}, respectively.}
	\label{tbl:results}
	\centering
	\small
	\begin{minipage}{\columnwidth}
		\centering
		\setlength{\tabcolsep}{4.5pt}
		\begin{tabular}{l|ccc|c|c}
			\toprule
			\multicolumn{6}{c}{D-NeRF~\cite{dnerf} (monocular, synthetic, 800$\times$800)} \\
			\midrule
			Method & PSNR $\uparrow$ & SSIM $\uparrow$ & LPIPS $\downarrow$ & Train Time $\downarrow$ & FPS $\uparrow$\\
			\midrule
			D-NeRF & 29.68 & 0.947 & 0.058 & 48\hour & $<$1\\
			TiNeuVox & 32.74 & 0.972 & 0.051 & 28\minute & 1.5\\
			K-Planes & 31.52 & 0.967 & 0.047 & 52\minute & 0.97\\
			HexPlane & 31.04 & 0.97 & 0.04 & \bf{11.5\minute} & 2.5\\ 
			4DGS & 35.34 & 0.985 & 0.021 & \underline{20\minute} & \underline{82}\\
			SC-GS & \underline{40.26} & \underline{0.992} & \underline{0.009} & 30\minute & \bf{164}\\
			\cmidrule{1-6}
			\rowcolor{gray!10} Ours & \bf{40.91} & \bf{0.995} & \bf{0.008} & 24\minute & 62\\
			\midrule
			
			\multicolumn{6}{c}{HyperNeRF~\cite{hypernerf} (monocular, real-world, 536$\times$960)} \\
			\midrule
			Method & PSNR $\uparrow$ & MS-SSIM $\uparrow$ & LPIPS $\downarrow$ & Train Time $\downarrow$ & FPS $\uparrow$\\
			\midrule
			HyperNeRF & 22.41 & 0.814 & 0.131 & 32\hour & $<$1\\
			TiNeuVox & 24.20 & 0.836 & 0.128 & \bf{30\minute} & 1\\ 
			4DGS & 25.24 & 0.845 & 0.116 & \underline{34\minute} & 32\\
			DeformGS & 25.02 & 0.822 & 0.116 & 1.5\hour & 13\\
			SaRO-GS & 25.38 & 0.850 & 0.110 & 1.2\hour & \underline{34}\\
			Grid4D & \underline{25.50} & \underline{0.856} & \underline{0.107} & 2.5\hour & \bf{37}\\
			\cmidrule{1-6}
			\rowcolor{gray!10} Ours & \bf{25.69} & \bf{0.858} & \bf{0.104} & 45\minute & 25\\
			\midrule
			
			\multicolumn{6}{c}{DyNeRF~\cite{dynerf} (multi-view, real-world, 1352$\times$1014)} \\
			\midrule
			Method & PSNR $\uparrow$ & D-SSIM $\downarrow$ & LPIPS $\downarrow$ & Train Time $\downarrow$ & FPS $\uparrow$\\
			\midrule
			DyNeRF & 29.58 & 0.020 & 0.083 & 1344\hour & $<$1\\
			HexPlane & 31.70 & 0.014 & 0.075 & 12\hour & 0.2\\
			4DGS & 31.17 & 0.016 & 0.049 & \underline{42\minute} & 30\\
			STG & 32.05 & 0.014 & 0.044 & 10\hour & \underline{110}\\
			SaRO-GS & 32.15 & 0.014 & 0.044 & 1.5\hour & 32\\
			Swift4D & \underline{32.23} & \underline{0.014} & \underline{0.043} & \bf{25\minute} & \bf{125}\\
			\cmidrule{1-6}
			\rowcolor{gray!10} Ours & \bf{32.35} & \bf{0.013} & \bf{0.042} & 50\minute & 21\\
			\bottomrule
		\end{tabular}
	\end{minipage}
\end{table}

\subsection{Experimental Setup}

\noindent \textbf{Datasets.} We evaluate on both monocular and multi-view dynamic scene datasets, including synthetic and real-world scenes:

\emph{Synthetic Dataset.} We use D-NeRF~\cite{dnerf} for synthetic scene evaluation rendering at 800$\times$800 resolution.

\emph{Real-World Datasets.} We use HyperNeRF~\cite{hypernerf} and DyNeRF~\cite{dynerf} for real-world scene evaluation. Camera poses are estimated using COLMAP~\cite{colmap}. We report quantitative results on HyperNeRF's ``vrig'' scenes at 536$\times$960 resolution and DyNeRF at 1352$\times$1014 resolution.

\noindent \textbf{Hyperparameters.} Our hyperparameter settings largely follow 4DGS~\cite{4dgs}. More experimental implementation details are summarized in Appendix~\ref{sec:detail}.

\noindent \textbf{Loss Function.} Following previous work~\cite{tineuvox, gaussian-splatting, 4dgs}, we supervise training using color loss $\mathcal{L}_{1}$ and structure loss $\mathcal{L}_{DSSIM}$ between the rendered and the ground-truth images. Additionally, a grid-based total variational loss~\cite{4dgs, hexplane} $\mathcal{L}_{TV}$ is applied:
\begin{equation}
	\mathcal{L}_{total} = (1-\lambda) \mathcal{L}_{1} + \lambda \mathcal{L}_{DSSIM} + \mathcal{L}_{TV},
\end{equation}
where $\lambda$, weighting between $\mathcal{L}_{1}$ and $\mathcal{L}_{DSSIM}$, is set to 0.2.

\subsection{Applications: Streaming}
To evaluate the practical utility of NeHaD in graphics applications, we extend it to adaptive streaming for VR requirements, which is challenging in bandwidth-constrained environments. We adapt NeHaD to streaming through two enhancements: (a) scale-aware anisotropic MipMapping for efficient multi-level texture sampling, and (b) layered progressive optimization for global level-of-detail (LOD) rendering. More details of streaming methodology are summarized in Appendix~\ref{sec:stream}. Qualitative results for adaptive streaming are illustrated in Fig. \ref{fig:streaming-results}.

\subsection{Experimental Results}
We compare NeHaD with state-of-the-art models on three datasets: the synthetic D-NeRF, the monocular\footnote{Monocular here refers to having only one viewpoint at any given time.} real-world HyperNeRF, and the multi-view real-world DyNeRF.

\noindent \textbf{Qualitative Comparisons.} Qualitative results for D-NeRF, HyperNeRF, and DyNeRF datasets are shown in Fig. \ref{fig:trajectory}, Fig. \ref{fig:dnerf}, Fig. \ref{fig:hypernerf}, and Fig. \ref{fig:dynerf}, respectively. Our method consistently outperforms baseline models in visual perception, delivering realistic renderings with coherent dynamic motion.

\noindent \textbf{Quantitative Comparisons.} Quantitative results are summarized in Tab. \ref{tbl:results}. NeHaD outperforms state-of-the-art methods across all quality evaluation metrics. While not matching the rendering speed of the most efficient methods like \cite{spacetime-gaussians} in real-world scenes, NeHaD maintains over 20 FPS, which is acceptable given its improved visual quality.

\subsection{Ablation Study}
We conduct ablation studies on the synthetic D-NeRF dataset. Tab. \ref{tbl:ablation} and Fig. \ref{fig:ablation} present our ablation results, with both quantitative and qualitative evaluations confirming the effectiveness of each proposed module in NeHaD.

\noindent \textbf{Ablation of Hamiltonian Neural Network Decoder.} The baseline without HNN decoder produces blurred and over-smoothed rendering results, likely due to low-rank assumptions and inadequate dynamics modeling. Detailed examination reveals Gaussian primitives positioned irregularly across surfaces, indicating the deformation field lacks physics-based inductive biases for unsupervised motion learning. Incorporating the HNN decoder improves visual quality metrics by approximately 6.7\% over the baseline.

\noindent \textbf{Ablation of Boltzmann Equilibrium Decomposition.} Without BED to separate static and dynamic Gaussians, rendering results exhibit noticeable artifacts in motion regions with mixed textures, significantly reducing quality metrics.

\noindent \textbf{Ablation of Physics-Informed Constraints.} Quantitative results confirm the effectiveness of physical constraints. Similar to findings in \cite{4dgs}, position deformation contributes most significantly to overall Gaussian deformation, with its physics regularization providing guidance that reduces motion trajectory ambiguity. Meanwhile, local rigidity regularization on rotation enhances motion stability through constrained rotation magnitudes.

\begin{table}[t]
	\caption{\textbf{Ablation study.} We start with 4DGS~\cite{4dgs} and improve it with proposed components on the synthetic dataset, respectively.}
	\label{tbl:ablation}
	\centering
	\small
	\setlength{\tabcolsep}{6pt}
	\begin{tabular}{rl|ccc}
		\toprule
		& Model configuration & PSNR $\uparrow$ & SSIM $\uparrow$ & LPIPS $\downarrow$ \\
		\midrule
		& Baseline model & 35.34 & 0.985 & 0.021\\
		\midrule
		+ & Hamiltonian neural network & 37.69 & 0.989 & 0.012\\
		\midrule
		\multirow{2.5}{*}{\rotatebox[origin=c]{90}{+ BED}} & Spatial-temporal decomposition & 37.46 & 0.989 & 0.014\\ \cmidrule{2-5}
		& Temporal-only decomposition & 36.08 & 0.986 & 0.019\\
		\midrule
		\multirow{2.5}{*}{\rotatebox[origin=c]{90}{+ PIC}} & Symplectic integration & 36.53 & 0.986 & 0.018\\ \cmidrule{2-5}
		& Local rigidity regularization & 36.04 & 0.986 & 0.020\\
		\midrule
		\rowcolor{gray!10} & Proposed model & 40.91 & 0.995 & 0.008\\
		\bottomrule
	\end{tabular}
\end{table}

\begin{figure}[t]
	\centering
	\includegraphics[width=1\columnwidth]{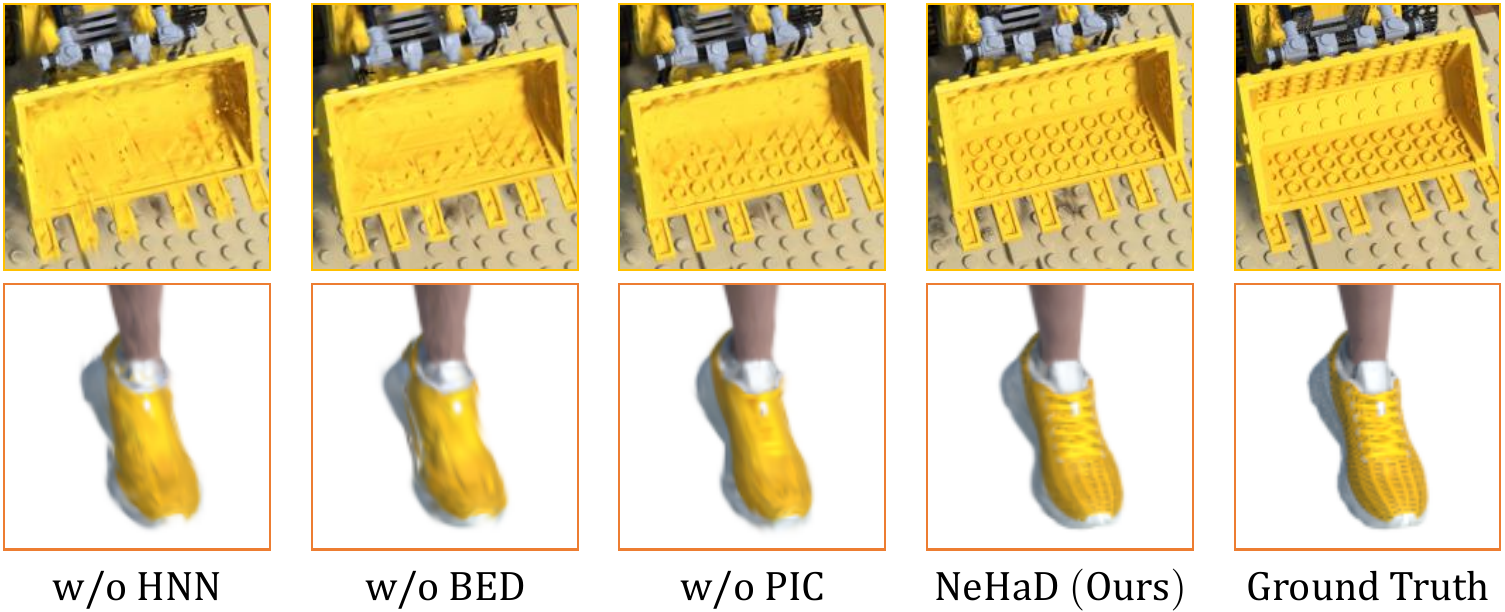}
	\caption{\textbf{Ablation visualization.} Without the proposed modules, rendering results exhibit motion artifacts and visual distortions. In contrast, our complete NeHaD model produces significantly higher quality results, demonstrating the effectiveness of our approach.}
	\label{fig:ablation}
\end{figure}

\begin{figure}[t]
	\centering
	\includegraphics[width=1\columnwidth]{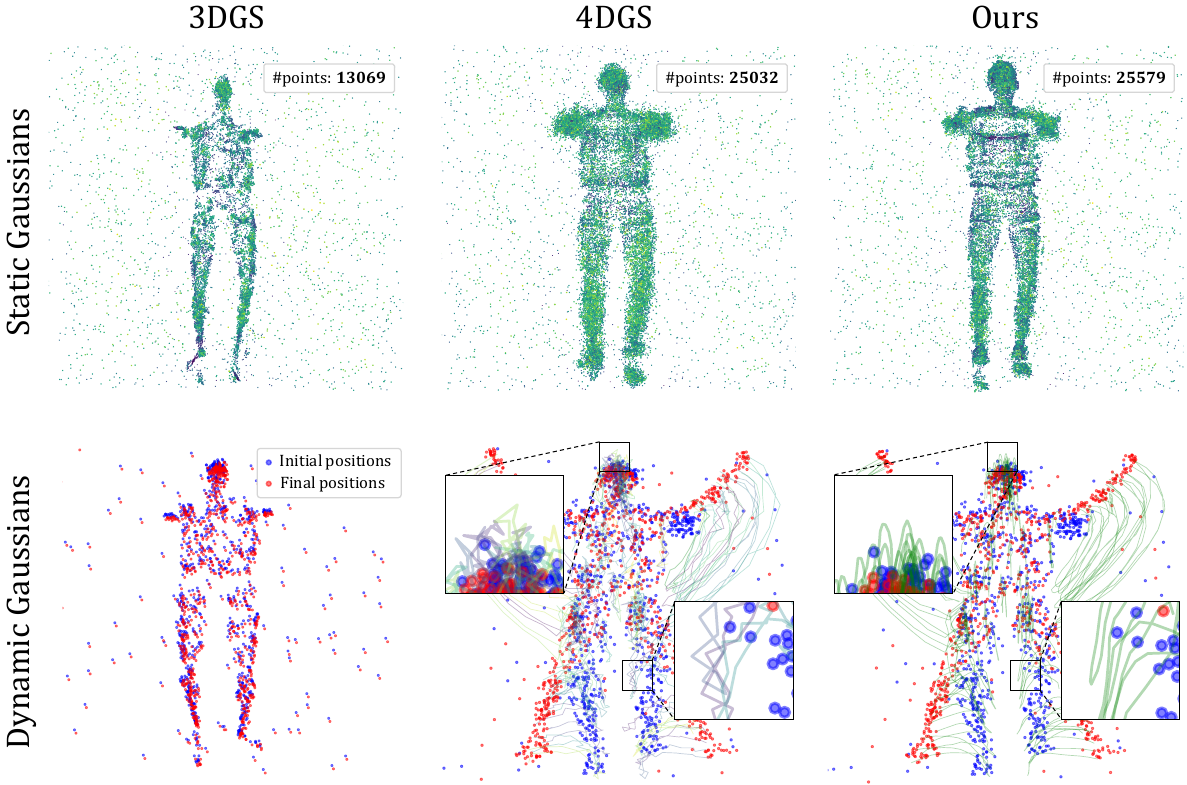}
	\caption{\textbf{Visualization of Gaussian spatial distributions and motion trajectories during deformation.} Compared to baseline~\cite{4dgs} using MLP-based decoders, our HNN-based approach inherently respects Hamiltonian mechanics principles, resulting in more directed, ordered, and natural movements instead of relying solely on stochastic optimization.}
	\label{fig:trajectory}
\end{figure}

\begin{figure*}[t]
	\centering
	\includegraphics[width=0.98\textwidth]{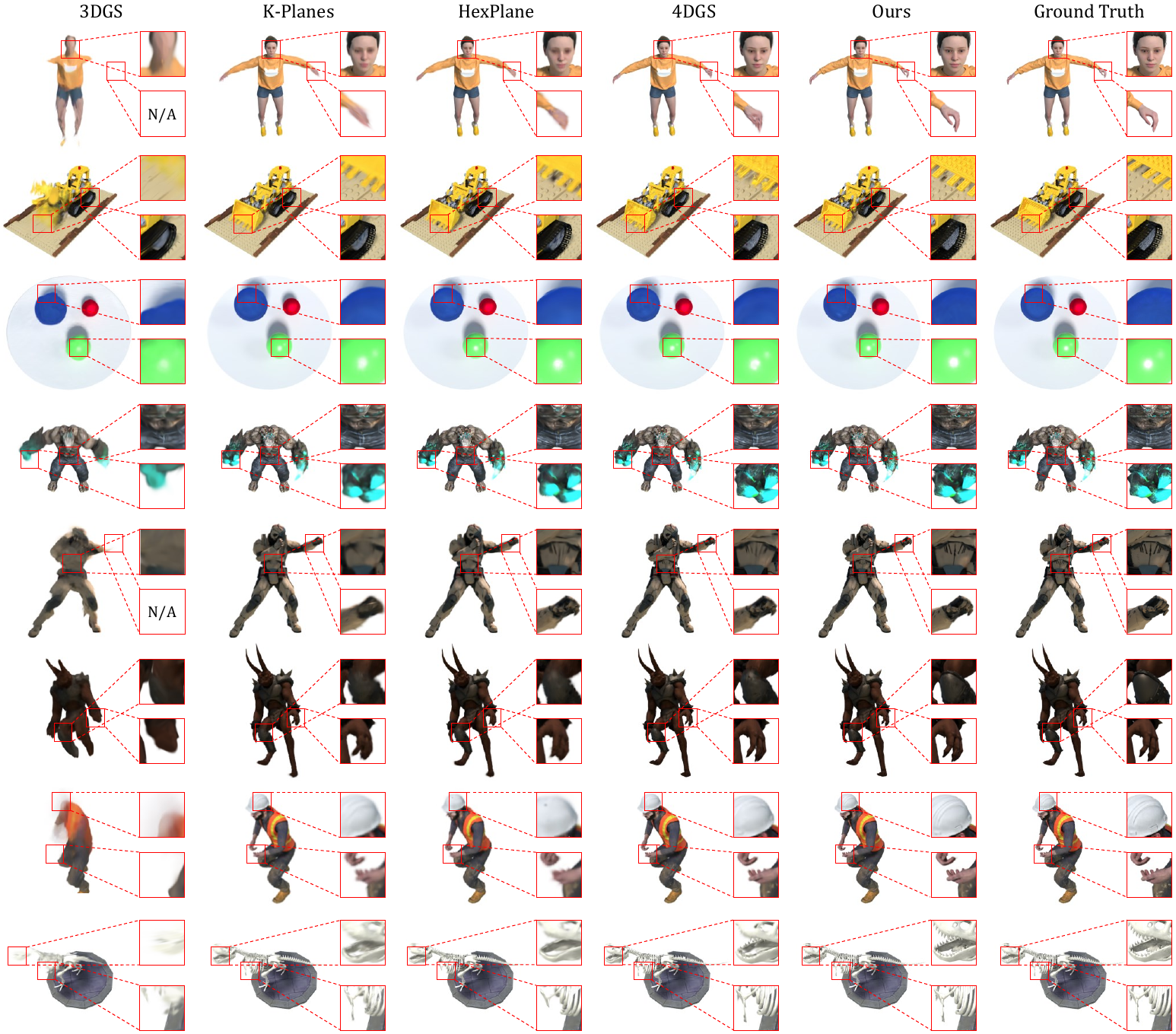}
	\caption{Qualitative results on D-NeRF~\cite{dnerf} dataset (N/A: not available).}
	\label{fig:dnerf}
\end{figure*}

\begin{figure*}[t]
	\centering
	\includegraphics[width=0.98\textwidth]{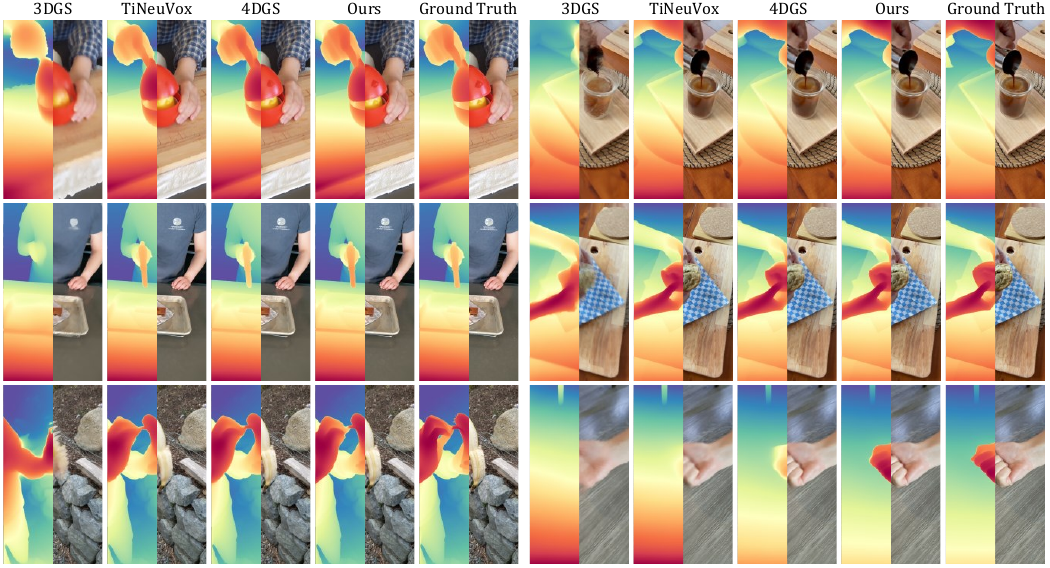}
	\caption{Qualitative results on HyperNeRF~\cite{hypernerf} dataset (depth maps are visualized for better comparisons).}
	\label{fig:hypernerf}
\end{figure*}

\begin{figure*}[t]
	\centering
	\includegraphics[width=0.98\textwidth]{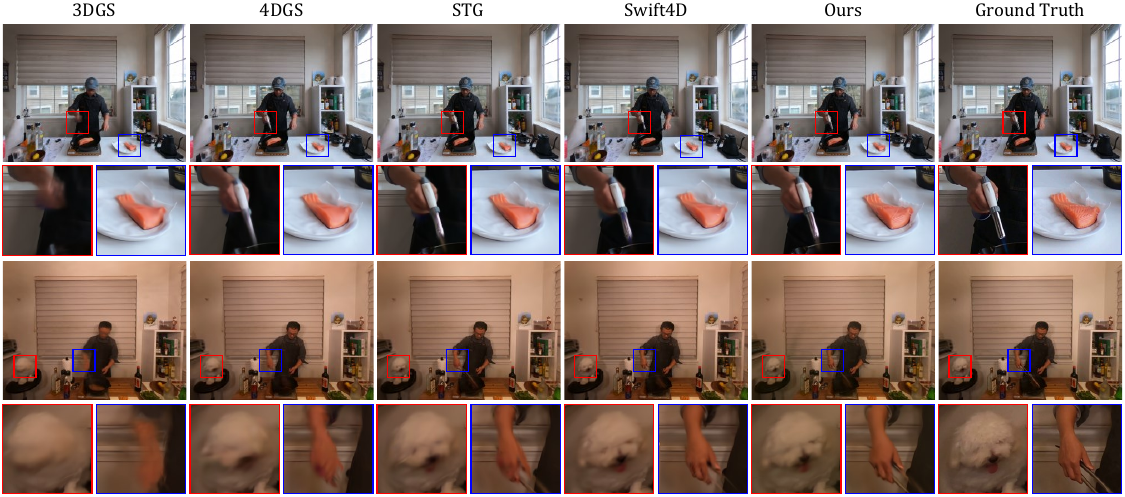}
	\caption{Qualitative results on DyNeRF~\cite{dynerf} dataset (\textit{flame salmon} and \textit{sear steak} scenes).}
	\label{fig:dynerf}
\end{figure*}

\section{Conclusion}
In this paper, we introduce NeHaD, a neural deformation field for dynamic Gaussian Splatting guided by Hamiltonian mechanics. Our method decomposes 4D spatial-temporal representations using Boltzmann equilibrium decomposition and employs an HNN-based decoder to predict Gaussian deformations with physics-informed constraints that ensure robustness against scene perturbations. The approach extends naturally to bandwidth-constrained streaming applications with minimal modifications. Extensive experiments demonstrate that our method achieves improved rendering reality.

\noindent \textbf{Limitation.} While improvements achieved, our method also has limitations. First, specialized regularization for real-world scenes may conflict with ground-truth dynamics, especially under occlusion, large dissipation, or fluid deformation. Second, learning Hamiltonian introduces extra computational overhead, leading to reduced real-time rendering performance. Future work could explore more flexible neural operators and efficient rendering pipelines.

\begin{acks}
The authors thank all the reviewers for their useful suggestions. This work is supported in part by the National Key Research and Development Program of China under Grant 2024YFF0509700, in part by the National Natural Science Foundation of China under Grant 62371063, Grant	62293481, Grant 62321001, Grant 92267301, in part by the Beijing Municipal Natural Science Foundation under Grant L232047, and sponsored by Beijing Nova Program.
\end{acks}

\bibliographystyle{ACM-Reference-Format}
\bibliography{reference}


\begin{thebibliography}{59}


\ifx \showCODEN    \undefined \def \showCODEN     #1{\unskip}     \fi
\ifx \showDOI      \undefined \def \showDOI       #1{#1}\fi
\ifx \showISBNx    \undefined \def \showISBNx     #1{\unskip}     \fi
\ifx \showISBNxiii \undefined \def \showISBNxiii  #1{\unskip}     \fi
\ifx \showISSN     \undefined \def \showISSN      #1{\unskip}     \fi
\ifx \showLCCN     \undefined \def \showLCCN      #1{\unskip}     \fi
\ifx \shownote     \undefined \def \shownote      #1{#1}          \fi
\ifx \showarticletitle \undefined \def \showarticletitle #1{#1}   \fi
\ifx \showURL      \undefined \def \showURL       {\relax}        \fi
\providecommand\bibfield[2]{#2}
\providecommand\bibinfo[2]{#2}
\providecommand\natexlab[1]{#1}
\providecommand\showeprint[2][]{arXiv:#2}

\bibitem[Cai et~al\mbox{.}(2022)]%
        {ndr}
\bibfield{author}{\bibinfo{person}{Hongrui Cai}, \bibinfo{person}{Wanquan
  Feng}, \bibinfo{person}{Xuetao Feng}, \bibinfo{person}{Yan Wang}, {and}
  \bibinfo{person}{Juyong Zhang}.} \bibinfo{year}{2022}\natexlab{}.
\newblock \showarticletitle{Neural surface reconstruction of dynamic scenes
  with monocular {RGB-D} camera}.
\newblock \bibinfo{journal}{\emph{Advances in Neural Information Processing
  Systems (NeurIPS)}}  \bibinfo{volume}{35} (\bibinfo{year}{2022}),
  \bibinfo{pages}{967--981}.
\newblock


\bibitem[Cao and Johnson(2023)]%
        {hexplane}
\bibfield{author}{\bibinfo{person}{Ang Cao} {and} \bibinfo{person}{Justin
  Johnson}.} \bibinfo{year}{2023}\natexlab{}.
\newblock \showarticletitle{{HexPlane}: A fast representation for dynamic
  scenes}. In \bibinfo{booktitle}{\emph{Proceedings of the IEEE/CVF Conference
  on Computer Vision and Pattern Recognition (CVPR)}}.
  \bibinfo{pages}{130--141}.
\newblock


\bibitem[Chan et~al\mbox{.}(2021)]%
        {eg3d}
\bibfield{author}{\bibinfo{person}{Eric~R. Chan}, \bibinfo{person}{Connor~Z.
  Lin}, \bibinfo{person}{Matthew~A. Chan}, \bibinfo{person}{Koki Nagano},
  \bibinfo{person}{Boxiao Pan}, \bibinfo{person}{Shalini De~Mello},
  \bibinfo{person}{Orazio Gallo}, \bibinfo{person}{Leonidas Guibas},
  \bibinfo{person}{Jonathan Tremblay}, \bibinfo{person}{Sameh Khamis},
  \bibinfo{person}{Tero Karras}, {and} \bibinfo{person}{Gordon Wetzstein}.}
  \bibinfo{year}{2021}\natexlab{}.
\newblock \showarticletitle{Efficient Geometry-aware {3D} Generative
  Adversarial Networks}. In \bibinfo{booktitle}{\emph{arXiv}}.
\newblock


\bibitem[Duan et~al\mbox{.}(2024)]%
        {beida4dgs}
\bibfield{author}{\bibinfo{person}{Yuanxing Duan}, \bibinfo{person}{Fangyin
  Wei}, \bibinfo{person}{Qiyu Dai}, \bibinfo{person}{Yuhang He},
  \bibinfo{person}{Wenzheng Chen}, {and} \bibinfo{person}{Baoquan Chen}.}
  \bibinfo{year}{2024}\natexlab{}.
\newblock \showarticletitle{{4D Gaussian Splatting}: Towards Efficient Novel
  View Synthesis for Dynamic Scenes}.
\newblock \bibinfo{journal}{\emph{arXiv preprint arXiv:2402.03307}}
  (\bibinfo{year}{2024}).
\newblock


\bibitem[Duisterhof et~al\mbox{.}(2023)]%
        {md-splatting}
\bibfield{author}{\bibinfo{person}{Bardienus~P Duisterhof},
  \bibinfo{person}{Zhao Mandi}, \bibinfo{person}{Yunchao Yao},
  \bibinfo{person}{Jia-Wei Liu}, \bibinfo{person}{Mike~Zheng Shou},
  \bibinfo{person}{Shuran Song}, {and} \bibinfo{person}{Jeffrey Ichnowski}.}
  \bibinfo{year}{2023}\natexlab{}.
\newblock \showarticletitle{{MD-Splatting}: Learning metric deformation from
  {4D} {Gaussians} in highly deformable scenes}.
\newblock \bibinfo{journal}{\emph{arXiv preprint arXiv:2312.00583}}
  (\bibinfo{year}{2023}).
\newblock


\bibitem[Fan et~al\mbox{.}(2024)]%
        {lightgaussian}
\bibfield{author}{\bibinfo{person}{Zhiwen Fan}, \bibinfo{person}{Kevin Wang},
  \bibinfo{person}{Kairun Wen}, \bibinfo{person}{Zehao Zhu},
  \bibinfo{person}{Dejia Xu}, \bibinfo{person}{Zhangyang Wang},
  {et~al\mbox{.}}} \bibinfo{year}{2024}\natexlab{}.
\newblock \showarticletitle{Lightgaussian: Unbounded 3d gaussian compression
  with 15x reduction and 200+ fps}.
\newblock \bibinfo{journal}{\emph{Advances in Neural Information Processing
  Systems (NeurIPS)}}  \bibinfo{volume}{37} (\bibinfo{year}{2024}),
  \bibinfo{pages}{140138--140158}.
\newblock


\bibitem[Fang et~al\mbox{.}(2022)]%
        {tineuvox}
\bibfield{author}{\bibinfo{person}{Jiemin Fang}, \bibinfo{person}{Taoran Yi},
  \bibinfo{person}{Xinggang Wang}, \bibinfo{person}{Lingxi Xie},
  \bibinfo{person}{Xiaopeng Zhang}, \bibinfo{person}{Wenyu Liu},
  \bibinfo{person}{Matthias Nie{\ss}ner}, {and} \bibinfo{person}{Qi Tian}.}
  \bibinfo{year}{2022}\natexlab{}.
\newblock \showarticletitle{Fast dynamic radiance fields with time-aware neural
  voxels}. In \bibinfo{booktitle}{\emph{SIGGRAPH Asia 2022 Conference Papers}}.
  \bibinfo{pages}{1--9}.
\newblock


\bibitem[Fridovich-Keil et~al\mbox{.}(2023)]%
        {k-planes}
\bibfield{author}{\bibinfo{person}{Sara Fridovich-Keil},
  \bibinfo{person}{Giacomo Meanti}, \bibinfo{person}{Frederik~Rahb{\ae}k
  Warburg}, \bibinfo{person}{Benjamin Recht}, {and} \bibinfo{person}{Angjoo
  Kanazawa}.} \bibinfo{year}{2023}\natexlab{}.
\newblock \showarticletitle{{K-Planes}: Explicit radiance fields in space,
  time, and appearance}. In \bibinfo{booktitle}{\emph{Proceedings of the
  IEEE/CVF Conference on Computer Vision and Pattern Recognition (CVPR)}}.
  \bibinfo{pages}{12479--12488}.
\newblock


\bibitem[Gao et~al\mbox{.}(2021)]%
        {dynamic-nerf}
\bibfield{author}{\bibinfo{person}{Chen Gao}, \bibinfo{person}{Ayush Saraf},
  \bibinfo{person}{Johannes Kopf}, {and} \bibinfo{person}{Jia-Bin Huang}.}
  \bibinfo{year}{2021}\natexlab{}.
\newblock \showarticletitle{Dynamic view synthesis from dynamic monocular
  video}. In \bibinfo{booktitle}{\emph{Proceedings of the IEEE/CVF
  International Conference on Computer Vision (ICCV)}}.
  \bibinfo{pages}{5712--5721}.
\newblock


\bibitem[Greydanus et~al\mbox{.}(2019)]%
        {hnn}
\bibfield{author}{\bibinfo{person}{Samuel Greydanus}, \bibinfo{person}{Misko
  Dzamba}, {and} \bibinfo{person}{Jason Yosinski}.}
  \bibinfo{year}{2019}\natexlab{}.
\newblock \showarticletitle{Hamiltonian neural networks}.
\newblock \bibinfo{journal}{\emph{Advances in Neural Information Processing
  Systems (NeurIPS)}}  \bibinfo{volume}{32} (\bibinfo{year}{2019}).
\newblock


\bibitem[Griffiths(2023)]%
        {griffiths2023introduction}
\bibfield{author}{\bibinfo{person}{David~J Griffiths}.}
  \bibinfo{year}{2023}\natexlab{}.
\newblock \bibinfo{booktitle}{\emph{Introduction to electrodynamics}}.
\newblock \bibinfo{publisher}{Cambridge University Press}.
\newblock


\bibitem[Guo et~al\mbox{.}(2023)]%
        {forwardflowdnerf}
\bibfield{author}{\bibinfo{person}{Xiang Guo}, \bibinfo{person}{Jiadai Sun},
  \bibinfo{person}{Yuchao Dai}, \bibinfo{person}{Guanying Chen},
  \bibinfo{person}{Xiaoqing Ye}, \bibinfo{person}{Xiao Tan},
  \bibinfo{person}{Errui Ding}, \bibinfo{person}{Yumeng Zhang}, {and}
  \bibinfo{person}{Jingdong Wang}.} \bibinfo{year}{2023}\natexlab{}.
\newblock \showarticletitle{Forward flow for novel view synthesis of dynamic
  scenes}. In \bibinfo{booktitle}{\emph{Proceedings of the IEEE/CVF
  International Conference on Computer Vision (ICCV)}}.
  \bibinfo{pages}{16022--16033}.
\newblock


\bibitem[Guo et~al\mbox{.}(2024)]%
        {guo2024motion}
\bibfield{author}{\bibinfo{person}{Zhiyang Guo}, \bibinfo{person}{Wengang
  Zhou}, \bibinfo{person}{Li Li}, \bibinfo{person}{Min Wang}, {and}
  \bibinfo{person}{Houqiang Li}.} \bibinfo{year}{2024}\natexlab{}.
\newblock \showarticletitle{Motion-aware 3d gaussian splatting for efficient
  dynamic scene reconstruction}.
\newblock \bibinfo{journal}{\emph{IEEE Transactions on Circuits and Systems for
  Video Technology}} (\bibinfo{year}{2024}).
\newblock


\bibitem[Huang et~al\mbox{.}(2023)]%
        {sc-gs}
\bibfield{author}{\bibinfo{person}{Yi-Hua Huang}, \bibinfo{person}{Yang-Tian
  Sun}, \bibinfo{person}{Ziyi Yang}, \bibinfo{person}{Xiaoyang Lyu},
  \bibinfo{person}{Yan-Pei Cao}, {and} \bibinfo{person}{Xiaojuan Qi}.}
  \bibinfo{year}{2023}\natexlab{}.
\newblock \showarticletitle{{SC-GS}: Sparse-Controlled {Gaussian} Splatting for
  Editable Dynamic Scenes}.
\newblock \bibinfo{journal}{\emph{arXiv preprint arXiv:2312.14937}}
  (\bibinfo{year}{2023}).
\newblock


\bibitem[Kerbl et~al\mbox{.}(2023)]%
        {gaussian-splatting}
\bibfield{author}{\bibinfo{person}{Bernhard Kerbl}, \bibinfo{person}{Georgios
  Kopanas}, \bibinfo{person}{Thomas Leimk{\"u}hler}, {and}
  \bibinfo{person}{George Drettakis}.} \bibinfo{year}{2023}\natexlab{}.
\newblock \showarticletitle{{3D} {Gaussian} Splatting for Real-Time Radiance
  Field Rendering}.
\newblock \bibinfo{journal}{\emph{ACM Transactions on Graphics (TOG)}}
  \bibinfo{volume}{42}, \bibinfo{number}{4} (\bibinfo{year}{2023}).
\newblock


\bibitem[Kingma and Ba(2014)]%
        {adam}
\bibfield{author}{\bibinfo{person}{Diederik~P Kingma} {and}
  \bibinfo{person}{Jimmy Ba}.} \bibinfo{year}{2014}\natexlab{}.
\newblock \showarticletitle{Adam: A method for stochastic optimization}.
\newblock \bibinfo{journal}{\emph{arXiv preprint arXiv:1412.6980}}
  (\bibinfo{year}{2014}).
\newblock


\bibitem[Kratimenos et~al\mbox{.}(2023)]%
        {dynmf}
\bibfield{author}{\bibinfo{person}{Agelos Kratimenos}, \bibinfo{person}{Jiahui
  Lei}, {and} \bibinfo{person}{Kostas Daniilidis}.}
  \bibinfo{year}{2023}\natexlab{}.
\newblock \showarticletitle{{DynMF}: Neural motion factorization for real-time
  dynamic view synthesis with {3D} {Gaussian} splatting}.
\newblock \bibinfo{journal}{\emph{arXiv preprint arXiv:2312.00112}}
  (\bibinfo{year}{2023}).
\newblock


\bibitem[Laine et~al\mbox{.}(2020)]%
        {laine2020modular}
\bibfield{author}{\bibinfo{person}{Samuli Laine}, \bibinfo{person}{Janne
  Hellsten}, \bibinfo{person}{Tero Karras}, \bibinfo{person}{Yeongho Seol},
  \bibinfo{person}{Jaakko Lehtinen}, {and} \bibinfo{person}{Timo Aila}.}
  \bibinfo{year}{2020}\natexlab{}.
\newblock \showarticletitle{Modular primitives for high-performance
  differentiable rendering}.
\newblock \bibinfo{journal}{\emph{ACM Transactions on Graphics (TOG)}}
  \bibinfo{volume}{39}, \bibinfo{number}{6} (\bibinfo{year}{2020}),
  \bibinfo{pages}{1--14}.
\newblock


\bibitem[Li et~al\mbox{.}(2022)]%
        {dynerf}
\bibfield{author}{\bibinfo{person}{Tianye Li}, \bibinfo{person}{Mira
  Slavcheva}, \bibinfo{person}{Michael Zollhoefer}, \bibinfo{person}{Simon
  Green}, \bibinfo{person}{Christoph Lassner}, \bibinfo{person}{Changil Kim},
  \bibinfo{person}{Tanner Schmidt}, \bibinfo{person}{Steven Lovegrove},
  \bibinfo{person}{Michael Goesele}, \bibinfo{person}{Richard Newcombe},
  {et~al\mbox{.}}} \bibinfo{year}{2022}\natexlab{}.
\newblock \showarticletitle{Neural 3d video synthesis from multi-view video}.
  In \bibinfo{booktitle}{\emph{Proceedings of the IEEE/CVF Conference on
  Computer Vision and Pattern Recognition (CVPR)}}.
  \bibinfo{pages}{5521--5531}.
\newblock


\bibitem[Li et~al\mbox{.}(2023)]%
        {spacetime-gaussians}
\bibfield{author}{\bibinfo{person}{Zhan Li}, \bibinfo{person}{Zhang Chen},
  \bibinfo{person}{Zhong Li}, {and} \bibinfo{person}{Yi Xu}.}
  \bibinfo{year}{2023}\natexlab{}.
\newblock \showarticletitle{Spacetime {Gaussian} Feature Splatting for
  Real-Time Dynamic View Synthesis}.
\newblock \bibinfo{journal}{\emph{arXiv preprint arXiv:2312.16812}}
  (\bibinfo{year}{2023}).
\newblock


\bibitem[Li et~al\mbox{.}(2021)]%
        {li2021neural}
\bibfield{author}{\bibinfo{person}{Zhengqi Li}, \bibinfo{person}{Simon
  Niklaus}, \bibinfo{person}{Noah Snavely}, {and} \bibinfo{person}{Oliver
  Wang}.} \bibinfo{year}{2021}\natexlab{}.
\newblock \showarticletitle{Neural scene flow fields for space-time view
  synthesis of dynamic scenes}. In \bibinfo{booktitle}{\emph{Proceedings of the
  IEEE/CVF Conference on Computer Vision and Pattern Recognition (CVPR)}}.
  \bibinfo{pages}{6498--6508}.
\newblock


\bibitem[Liang et~al\mbox{.}(2023a)]%
        {gaufre}
\bibfield{author}{\bibinfo{person}{Yiqing Liang}, \bibinfo{person}{Numair
  Khan}, \bibinfo{person}{Zhengqin Li}, \bibinfo{person}{Thu Nguyen-Phuoc},
  \bibinfo{person}{Douglas Lanman}, \bibinfo{person}{James Tompkin}, {and}
  \bibinfo{person}{Lei Xiao}.} \bibinfo{year}{2023}\natexlab{a}.
\newblock \showarticletitle{{GauFRe}: {Gaussian} Deformation Fields for
  Real-time Dynamic Novel View Synthesis}.
\newblock \bibinfo{journal}{\emph{arXiv preprint arXiv:2312.11458}}
  (\bibinfo{year}{2023}).
\newblock


\bibitem[Liang et~al\mbox{.}(2023b)]%
        {saff}
\bibfield{author}{\bibinfo{person}{Yiqing Liang}, \bibinfo{person}{Eliot
  Laidlaw}, \bibinfo{person}{Alexander Meyerowitz}, \bibinfo{person}{Srinath
  Sridhar}, {and} \bibinfo{person}{James Tompkin}.}
  \bibinfo{year}{2023}\natexlab{b}.
\newblock \showarticletitle{Semantic attention flow fields for monocular
  dynamic scene decomposition}. In \bibinfo{booktitle}{\emph{Proceedings of the
  IEEE/CVF International Conference on Computer Vision (ICCV)}}.
  \bibinfo{pages}{21797--21806}.
\newblock


\bibitem[Lin et~al\mbox{.}(2023)]%
        {gaussian-flow}
\bibfield{author}{\bibinfo{person}{Youtian Lin}, \bibinfo{person}{Zuozhuo Dai},
  \bibinfo{person}{Siyu Zhu}, {and} \bibinfo{person}{Yao Yao}.}
  \bibinfo{year}{2023}\natexlab{}.
\newblock \showarticletitle{{Gaussian-Flow}: {4D} reconstruction with dynamic
  {3D Gaussian} particle}.
\newblock \bibinfo{journal}{\emph{arXiv preprint arXiv:2312.03431}}
  (\bibinfo{year}{2023}).
\newblock


\bibitem[Liu et~al\mbox{.}(2023)]%
        {rodynrf}
\bibfield{author}{\bibinfo{person}{Yu-Lun Liu}, \bibinfo{person}{Chen Gao},
  \bibinfo{person}{Andreas Meuleman}, \bibinfo{person}{Hung-Yu Tseng},
  \bibinfo{person}{Ayush Saraf}, \bibinfo{person}{Changil Kim},
  \bibinfo{person}{Yung-Yu Chuang}, \bibinfo{person}{Johannes Kopf}, {and}
  \bibinfo{person}{Jia-Bin Huang}.} \bibinfo{year}{2023}\natexlab{}.
\newblock \showarticletitle{Robust dynamic radiance fields}. In
  \bibinfo{booktitle}{\emph{Proceedings of the IEEE/CVF Conference on Computer
  Vision and Pattern Recognition (CVPR)}}. \bibinfo{pages}{13--23}.
\newblock


\bibitem[Lu et~al\mbox{.}(2024)]%
        {gags}
\bibfield{author}{\bibinfo{person}{Zhicheng Lu}, \bibinfo{person}{Xiang Guo},
  \bibinfo{person}{Le Hui}, \bibinfo{person}{Tianrui Chen},
  \bibinfo{person}{Min Yang}, \bibinfo{person}{Xiao Tang},
  \bibinfo{person}{Feng Zhu}, {and} \bibinfo{person}{Yuchao Dai}.}
  \bibinfo{year}{2024}\natexlab{}.
\newblock \showarticletitle{{3D} Geometry-aware Deformable {Gaussian} Splatting
  for Dynamic View Synthesis}.
\newblock \bibinfo{journal}{\emph{arXiv preprint arXiv:2404.06270}}
  (\bibinfo{year}{2024}).
\newblock


\bibitem[Luiten et~al\mbox{.}(2023)]%
        {dynamic-3d-gaussians}
\bibfield{author}{\bibinfo{person}{Jonathon Luiten}, \bibinfo{person}{Georgios
  Kopanas}, \bibinfo{person}{Bastian Leibe}, {and} \bibinfo{person}{Deva
  Ramanan}.} \bibinfo{year}{2023}\natexlab{}.
\newblock \showarticletitle{{Dynamic 3D Gaussians}: Tracking by persistent
  dynamic view synthesis}.
\newblock \bibinfo{journal}{\emph{arXiv preprint arXiv:2308.09713}}
  (\bibinfo{year}{2023}).
\newblock


\bibitem[Mildenhall et~al\mbox{.}(2021)]%
        {nerf}
\bibfield{author}{\bibinfo{person}{Ben Mildenhall}, \bibinfo{person}{Pratul~P
  Srinivasan}, \bibinfo{person}{Matthew Tancik}, \bibinfo{person}{Jonathan~T
  Barron}, \bibinfo{person}{Ravi Ramamoorthi}, {and} \bibinfo{person}{Ren Ng}.}
  \bibinfo{year}{2021}\natexlab{}.
\newblock \showarticletitle{{NeRF}: Representing scenes as neural radiance
  fields for view synthesis}.
\newblock \bibinfo{journal}{\emph{Communications of the ACM (CACM)}}
  \bibinfo{volume}{65}, \bibinfo{number}{1} (\bibinfo{year}{2021}),
  \bibinfo{pages}{99--106}.
\newblock


\bibitem[M{\"u}ller et~al\mbox{.}(2022)]%
        {instant-ngp}
\bibfield{author}{\bibinfo{person}{Thomas M{\"u}ller}, \bibinfo{person}{Alex
  Evans}, \bibinfo{person}{Christoph Schied}, {and} \bibinfo{person}{Alexander
  Keller}.} \bibinfo{year}{2022}\natexlab{}.
\newblock \showarticletitle{Instant neural graphics primitives with a
  multiresolution hash encoding}.
\newblock \bibinfo{journal}{\emph{ACM Transactions on Graphics (TOG)}}
  \bibinfo{volume}{41}, \bibinfo{number}{4} (\bibinfo{year}{2022}),
  \bibinfo{pages}{1--15}.
\newblock


\bibitem[Noether(1971)]%
        {noether1971invariant}
\bibfield{author}{\bibinfo{person}{Emmy Noether}.}
  \bibinfo{year}{1971}\natexlab{}.
\newblock \showarticletitle{Invariant variation problems}.
\newblock \bibinfo{journal}{\emph{Transport Theory and Statistical Physics}}
  \bibinfo{volume}{1}, \bibinfo{number}{3} (\bibinfo{year}{1971}),
  \bibinfo{pages}{186--207}.
\newblock


\bibitem[Park et~al\mbox{.}(2021a)]%
        {nerfies}
\bibfield{author}{\bibinfo{person}{Keunhong Park}, \bibinfo{person}{Utkarsh
  Sinha}, \bibinfo{person}{Jonathan~T. Barron}, \bibinfo{person}{Sofien
  Bouaziz}, \bibinfo{person}{Dan~B Goldman}, \bibinfo{person}{Steven~M. Seitz},
  {and} \bibinfo{person}{Ricardo Martin-Brualla}.}
  \bibinfo{year}{2021}\natexlab{a}.
\newblock \showarticletitle{Nerfies: Deformable Neural Radiance Fields}. In
  \bibinfo{booktitle}{\emph{Proceedings of the IEEE/CVF International
  Conference on Computer Vision (ICCV)}}. \bibinfo{pages}{5865--5874}.
\newblock


\bibitem[Park et~al\mbox{.}(2021b)]%
        {hypernerf}
\bibfield{author}{\bibinfo{person}{Keunhong Park}, \bibinfo{person}{Utkarsh
  Sinha}, \bibinfo{person}{Peter Hedman}, \bibinfo{person}{Jonathan~T. Barron},
  \bibinfo{person}{Sofien Bouaziz}, \bibinfo{person}{Dan~B Goldman},
  \bibinfo{person}{Ricardo Martin-Brualla}, {and} \bibinfo{person}{Steven~M.
  Seitz}.} \bibinfo{year}{2021}\natexlab{b}.
\newblock \showarticletitle{HyperNeRF: a higher-dimensional representation for
  topologically varying neural radiance fields}.
\newblock \bibinfo{journal}{\emph{ACM Transactions on Graphics (TOG)}}
  \bibinfo{volume}{40}, \bibinfo{number}{6}, Article \bibinfo{articleno}{238}
  (\bibinfo{date}{dec} \bibinfo{year}{2021}), \bibinfo{numpages}{12}~pages.
\newblock


\bibitem[Pumarola et~al\mbox{.}(2021)]%
        {dnerf}
\bibfield{author}{\bibinfo{person}{Albert Pumarola}, \bibinfo{person}{Enric
  Corona}, \bibinfo{person}{Gerard Pons-Moll}, {and} \bibinfo{person}{Francesc
  Moreno-Noguer}.} \bibinfo{year}{2021}\natexlab{}.
\newblock \showarticletitle{D-nerf: Neural radiance fields for dynamic scenes}.
  In \bibinfo{booktitle}{\emph{Proceedings of the IEEE/CVF Conference on
  Computer Vision and Pattern Recognition (CVPR)}}.
  \bibinfo{pages}{10318--10327}.
\newblock


\bibitem[Schonberger and Frahm(2016)]%
        {colmap}
\bibfield{author}{\bibinfo{person}{Johannes~L Schonberger} {and}
  \bibinfo{person}{Jan-Michael Frahm}.} \bibinfo{year}{2016}\natexlab{}.
\newblock \showarticletitle{{Structure-from-Motion} revisited}. In
  \bibinfo{booktitle}{\emph{Proceedings of the IEEE Conference on Computer
  Vision and Pattern Recognition (CVPR)}}. \bibinfo{pages}{4104--4113}.
\newblock


\bibitem[Shao et~al\mbox{.}(2023)]%
        {tensor4d}
\bibfield{author}{\bibinfo{person}{Ruizhi Shao}, \bibinfo{person}{Zerong
  Zheng}, \bibinfo{person}{Hanzhang Tu}, \bibinfo{person}{Boning Liu},
  \bibinfo{person}{Hongwen Zhang}, {and} \bibinfo{person}{Yebin Liu}.}
  \bibinfo{year}{2023}\natexlab{}.
\newblock \showarticletitle{{Tensor4D}: Efficient neural {4D} decomposition for
  high-fidelity dynamic reconstruction and rendering}. In
  \bibinfo{booktitle}{\emph{Proceedings of the IEEE/CVF Conference on Computer
  Vision and Pattern Recognition (CVPR)}}. \bibinfo{pages}{16632--16642}.
\newblock


\bibitem[Shi et~al\mbox{.}(2024)]%
        {lapisgs}
\bibfield{author}{\bibinfo{person}{Yuang Shi}, \bibinfo{person}{G{\'e}raldine
  Morin}, \bibinfo{person}{Simone Gasparini}, {and} \bibinfo{person}{Wei~Tsang
  Ooi}.} \bibinfo{year}{2024}\natexlab{}.
\newblock \showarticletitle{Lapisgs: Layered progressive 3d gaussian splatting
  for adaptive streaming}.
\newblock \bibinfo{journal}{\emph{arXiv preprint arXiv:2408.14823}}
  (\bibinfo{year}{2024}).
\newblock


\bibitem[Sholokhov et~al\mbox{.}(2023)]%
        {sholokhov2023physics}
\bibfield{author}{\bibinfo{person}{Aleksei Sholokhov}, \bibinfo{person}{Yuying
  Liu}, \bibinfo{person}{Hassan Mansour}, {and} \bibinfo{person}{Saleh Nabi}.}
  \bibinfo{year}{2023}\natexlab{}.
\newblock \showarticletitle{Physics-informed neural ODE (PINODE): embedding
  physics into models using collocation points}.
\newblock \bibinfo{journal}{\emph{Scientific Reports}} \bibinfo{volume}{13},
  \bibinfo{number}{1} (\bibinfo{year}{2023}), \bibinfo{pages}{10166}.
\newblock


\bibitem[Somraj et~al\mbox{.}(2024)]%
        {factorized-motion}
\bibfield{author}{\bibinfo{person}{Nagabhushan Somraj}, \bibinfo{person}{Kapil
  Choudhary}, \bibinfo{person}{Sai~Harsha Mupparaju}, {and}
  \bibinfo{person}{Rajiv Soundararajan}.} \bibinfo{year}{2024}\natexlab{}.
\newblock \showarticletitle{Factorized Motion Fields for Fast Sparse Input
  Dynamic View Synthesis}.
\newblock \bibinfo{journal}{\emph{arXiv preprint arXiv:2404.11669}}
  (\bibinfo{year}{2024}).
\newblock


\bibitem[Song et~al\mbox{.}(2023)]%
        {nerfplayer}
\bibfield{author}{\bibinfo{person}{Liangchen Song}, \bibinfo{person}{Anpei
  Chen}, \bibinfo{person}{Zhong Li}, \bibinfo{person}{Zhang Chen},
  \bibinfo{person}{Lele Chen}, \bibinfo{person}{Junsong Yuan},
  \bibinfo{person}{Yi Xu}, {and} \bibinfo{person}{Andreas Geiger}.}
  \bibinfo{year}{2023}\natexlab{}.
\newblock \showarticletitle{{NeRFPlayer}: A streamable dynamic scene
  representation with decomposed neural radiance fields}.
\newblock \bibinfo{journal}{\emph{IEEE Transactions on Visualization and
  Computer Graphics (TVCG)}} \bibinfo{volume}{29}, \bibinfo{number}{5}
  (\bibinfo{year}{2023}), \bibinfo{pages}{2732--2742}.
\newblock


\bibitem[Sorkine and Alexa(2007)]%
        {arap}
\bibfield{author}{\bibinfo{person}{Olga Sorkine} {and} \bibinfo{person}{Marc
  Alexa}.} \bibinfo{year}{2007}\natexlab{}.
\newblock \showarticletitle{As-rigid-as-possible surface modeling}. In
  \bibinfo{booktitle}{\emph{Symposium on Geometry Processing (SGP)}},
  Vol.~\bibinfo{volume}{4}. Citeseer, \bibinfo{pages}{109--116}.
\newblock


\bibitem[Sun et~al\mbox{.}(2024a)]%
        {3dgstream}
\bibfield{author}{\bibinfo{person}{Jiakai Sun}, \bibinfo{person}{Han Jiao},
  \bibinfo{person}{Guangyuan Li}, \bibinfo{person}{Zhanjie Zhang},
  \bibinfo{person}{Lei Zhao}, {and} \bibinfo{person}{Wei Xing}.}
  \bibinfo{year}{2024}\natexlab{a}.
\newblock \showarticletitle{{3DGStream}: On-the-fly training of {3D Gaussians}
  for efficient streaming of photo-realistic free-viewpoint videos}.
\newblock \bibinfo{journal}{\emph{arXiv preprint arXiv:2403.01444}}
  (\bibinfo{year}{2024}).
\newblock


\bibitem[Sun et~al\mbox{.}(2024b)]%
        {sun2024multi}
\bibfield{author}{\bibinfo{person}{Yuan-Chun Sun}, \bibinfo{person}{Yuang Shi},
  \bibinfo{person}{Wei~Tsang Ooi}, \bibinfo{person}{Chun-Ying Huang}, {and}
  \bibinfo{person}{Cheng-Hsin Hsu}.} \bibinfo{year}{2024}\natexlab{b}.
\newblock \showarticletitle{Multi-frame bitrate allocation of dynamic 3D
  Gaussian splatting streaming over dynamic networks}. In
  \bibinfo{booktitle}{\emph{Proceedings of the 2024 SIGCOMM Workshop on
  Emerging Multimedia Systems}}. \bibinfo{pages}{1--7}.
\newblock


\bibitem[Süli and Mayers(2003)]%
        {ode-solver}
\bibfield{author}{\bibinfo{person}{Endre Süli} {and} \bibinfo{person}{David~F.
  Mayers}.} \bibinfo{year}{2003}\natexlab{}.
\newblock \bibinfo{booktitle}{\emph{An Introduction to Numerical Analysis}}.
\newblock \bibinfo{publisher}{Cambridge University Press}.
\newblock


\bibitem[Tretschk et~al\mbox{.}(2021)]%
        {nr-nerf}
\bibfield{author}{\bibinfo{person}{Edgar Tretschk}, \bibinfo{person}{Ayush
  Tewari}, \bibinfo{person}{Vladislav Golyanik}, \bibinfo{person}{Michael
  Zollh{\"o}fer}, \bibinfo{person}{Christoph Lassner}, {and}
  \bibinfo{person}{Christian Theobalt}.} \bibinfo{year}{2021}\natexlab{}.
\newblock \showarticletitle{{Non-Rigid Neural Radiance Fields}: Reconstruction
  and novel view synthesis of a dynamic scene from monocular video}. In
  \bibinfo{booktitle}{\emph{Proceedings of the IEEE/CVF International
  Conference on Computer Vision (ICCV)}}. \bibinfo{pages}{12959--12970}.
\newblock


\bibitem[Wang et~al\mbox{.}(2023b)]%
        {fsdnerf}
\bibfield{author}{\bibinfo{person}{Chaoyang Wang},
  \bibinfo{person}{Lachlan~Ewen MacDonald}, \bibinfo{person}{Laszlo~A Jeni},
  {and} \bibinfo{person}{Simon Lucey}.} \bibinfo{year}{2023}\natexlab{b}.
\newblock \showarticletitle{Flow supervision for Deformable {NeRF}}. In
  \bibinfo{booktitle}{\emph{Proceedings of the IEEE/CVF Conference on Computer
  Vision and Pattern Recognition (CVPR)}}. \bibinfo{pages}{21128--21137}.
\newblock


\bibitem[Wang et~al\mbox{.}(2023a)]%
        {masked-spacetime-hashing}
\bibfield{author}{\bibinfo{person}{Feng Wang}, \bibinfo{person}{Zilong Chen},
  \bibinfo{person}{Guokang Wang}, \bibinfo{person}{Yafei Song}, {and}
  \bibinfo{person}{Huaping Liu}.} \bibinfo{year}{2023}\natexlab{a}.
\newblock \showarticletitle{Masked Space-Time Hash Encoding for Efficient
  Dynamic Scene Reconstruction}.
\newblock \bibinfo{journal}{\emph{Advances in Neural Information Processing
  Systems (NeurIPS)}} (\bibinfo{year}{2023}).
\newblock


\bibitem[Wang et~al\mbox{.}(2004)]%
        {ssim}
\bibfield{author}{\bibinfo{person}{Zhou Wang}, \bibinfo{person}{Alan~C Bovik},
  \bibinfo{person}{Hamid~R Sheikh}, {and} \bibinfo{person}{Eero~P Simoncelli}.}
  \bibinfo{year}{2004}\natexlab{}.
\newblock \showarticletitle{Image quality assessment: from error visibility to
  structural similarity}.
\newblock \bibinfo{journal}{\emph{IEEE Transactions on Image Processing (TIP)}}
  \bibinfo{volume}{13}, \bibinfo{number}{4} (\bibinfo{year}{2004}),
  \bibinfo{pages}{600--612}.
\newblock


\bibitem[Wu et~al\mbox{.}(2024)]%
        {4dgs}
\bibfield{author}{\bibinfo{person}{Guanjun Wu}, \bibinfo{person}{Taoran Yi},
  \bibinfo{person}{Jiemin Fang}, \bibinfo{person}{Lingxi Xie},
  \bibinfo{person}{Xiaopeng Zhang}, \bibinfo{person}{Wei Wei},
  \bibinfo{person}{Wenyu Liu}, \bibinfo{person}{Qi Tian}, {and}
  \bibinfo{person}{Xinggang Wang}.} \bibinfo{year}{2024}\natexlab{}.
\newblock \showarticletitle{4D Gaussian Splatting for Real-Time Dynamic Scene
  Rendering}. In \bibinfo{booktitle}{\emph{Proceedings of the IEEE/CVF
  Conference on Computer Vision and Pattern Recognition (CVPR)}}.
  \bibinfo{pages}{20310--20320}.
\newblock


\bibitem[Wu et~al\mbox{.}(2025)]%
        {swift4d}
\bibfield{author}{\bibinfo{person}{Jiahao Wu}, \bibinfo{person}{Rui Peng},
  \bibinfo{person}{Zhiyan Wang}, \bibinfo{person}{Lu Xiao},
  \bibinfo{person}{Luyang Tang}, \bibinfo{person}{Jinbo Yan},
  \bibinfo{person}{Kaiqiang Xiong}, {and} \bibinfo{person}{Ronggang Wang}.}
  \bibinfo{year}{2025}\natexlab{}.
\newblock \showarticletitle{Swift4D: Adaptive divide-and-conquer Gaussian
  Splatting for compact and efficient reconstruction of dynamic scene}.
\newblock \bibinfo{journal}{\emph{arXiv preprint arXiv:2503.12307}}
  (\bibinfo{year}{2025}).
\newblock


\bibitem[Xu et~al\mbox{.}(2024a)]%
        {grid4d}
\bibfield{author}{\bibinfo{person}{Jiawei Xu}, \bibinfo{person}{Zexin Fan},
  \bibinfo{person}{Jian Yang}, {and} \bibinfo{person}{Jin Xie}.}
  \bibinfo{year}{2024}\natexlab{a}.
\newblock \showarticletitle{Grid4D: 4D Decomposed Hash Encoding for
  High-Fidelity Dynamic Gaussian Splatting}.
\newblock \bibinfo{journal}{\emph{arXiv preprint arXiv:2410.20815}}
  (\bibinfo{year}{2024}).
\newblock


\bibitem[Xu et~al\mbox{.}(2024b)]%
        {4k4d}
\bibfield{author}{\bibinfo{person}{Zhen Xu}, \bibinfo{person}{Sida Peng},
  \bibinfo{person}{Haotong Lin}, \bibinfo{person}{Guangzhao He},
  \bibinfo{person}{Jiaming Sun}, \bibinfo{person}{Yujun Shen},
  \bibinfo{person}{Hujun Bao}, {and} \bibinfo{person}{Xiaowei Zhou}.}
  \bibinfo{year}{2024}\natexlab{b}.
\newblock \showarticletitle{{4K4D}: Real-time {4D} view synthesis at {4K}
  resolution}. In \bibinfo{booktitle}{\emph{Proceedings of the IEEE/CVF
  Conference on Computer Vision and Pattern Recognition (CVPR)}}.
  \bibinfo{pages}{20029--20040}.
\newblock


\bibitem[Yan et~al\mbox{.}(2024)]%
        {saro-gs}
\bibfield{author}{\bibinfo{person}{Jinbo Yan}, \bibinfo{person}{Rui Peng},
  \bibinfo{person}{Luyang Tang}, {and} \bibinfo{person}{Ronggang Wang}.}
  \bibinfo{year}{2024}\natexlab{}.
\newblock \showarticletitle{4D Gaussian Splatting with Scale-aware Residual
  Field and Adaptive Optimization for Real-time rendering of temporally complex
  dynamic scenes}. In \bibinfo{booktitle}{\emph{Proceedings of the 32nd ACM
  International Conference on Multimedia (ACM MM)}}.
  \bibinfo{pages}{7871--7880}.
\newblock


\bibitem[Yan et~al\mbox{.}(2023)]%
        {nerf-ds}
\bibfield{author}{\bibinfo{person}{Zhiwen Yan}, \bibinfo{person}{Chen Li},
  {and} \bibinfo{person}{Gim~Hee Lee}.} \bibinfo{year}{2023}\natexlab{}.
\newblock \showarticletitle{{NeRF-DS}: Neural radiance fields for dynamic
  specular objects}. In \bibinfo{booktitle}{\emph{Proceedings of the IEEE/CVF
  Conference on Computer Vision and Pattern Recognition (CVPR)}}.
  \bibinfo{pages}{8285--8295}.
\newblock


\bibitem[Yang et~al\mbox{.}(2024)]%
        {deformable3d}
\bibfield{author}{\bibinfo{person}{Ziyi Yang}, \bibinfo{person}{Xinyu Gao},
  \bibinfo{person}{Wenming Zhou}, \bibinfo{person}{Shaohui Jiao},
  \bibinfo{person}{Yuqing Zhang}, {and} \bibinfo{person}{Xiaogang Jin}.}
  \bibinfo{year}{2024}\natexlab{}.
\newblock \showarticletitle{Deformable 3D Gaussians for High-Fidelity Monocular
  Dynamic Scene Reconstruction}.
\newblock \bibinfo{journal}{\emph{2024 IEEE/CVF Conference on Computer Vision
  and Pattern Recognition (CVPR)}} (\bibinfo{year}{2024}),
  \bibinfo{pages}{20331--20341}.
\newblock


\bibitem[Yang et~al\mbox{.}(2023)]%
        {fudan4dgs}
\bibfield{author}{\bibinfo{person}{Zeyu Yang}, \bibinfo{person}{Hongye Yang},
  \bibinfo{person}{Zijie Pan}, \bibinfo{person}{Xiatian Zhu}, {and}
  \bibinfo{person}{Li Zhang}.} \bibinfo{year}{2023}\natexlab{}.
\newblock \showarticletitle{Real-time photorealistic dynamic scene
  representation and rendering with {4D} {Gaussian} splatting}.
\newblock \bibinfo{journal}{\emph{arXiv preprint arXiv:2310.10642}}
  (\bibinfo{year}{2023}).
\newblock


\bibitem[Yu et~al\mbox{.}(2023)]%
        {cogs}
\bibfield{author}{\bibinfo{person}{Heng Yu}, \bibinfo{person}{Joel Julin},
  \bibinfo{person}{Zolt{\'a}n~{\'A} Milacski}, \bibinfo{person}{Koichiro
  Niinuma}, {and} \bibinfo{person}{L{\'a}szl{\'o}~A Jeni}.}
  \bibinfo{year}{2023}\natexlab{}.
\newblock \showarticletitle{{CoGS}: Controllable {Gaussian} splatting}.
\newblock \bibinfo{journal}{\emph{arXiv preprint arXiv:2312.05664}}
  (\bibinfo{year}{2023}).
\newblock


\bibitem[Yuan et~al\mbox{.}(2021)]%
        {star}
\bibfield{author}{\bibinfo{person}{Wentao Yuan}, \bibinfo{person}{Zhaoyang Lv},
  \bibinfo{person}{Tanner Schmidt}, {and} \bibinfo{person}{Steven Lovegrove}.}
  \bibinfo{year}{2021}\natexlab{}.
\newblock \showarticletitle{{STaR}: Self-supervised tracking and reconstruction
  of rigid objects in motion with neural rendering}. In
  \bibinfo{booktitle}{\emph{Proceedings of the IEEE/CVF Conference on Computer
  Vision and Pattern Recognition (CVPR)}}. \bibinfo{pages}{13144--13152}.
\newblock


\bibitem[Zhang et~al\mbox{.}(2018)]%
        {lpips}
\bibfield{author}{\bibinfo{person}{Richard Zhang}, \bibinfo{person}{Phillip
  Isola}, \bibinfo{person}{Alexei~A Efros}, \bibinfo{person}{Eli Shechtman},
  {and} \bibinfo{person}{Oliver Wang}.} \bibinfo{year}{2018}\natexlab{}.
\newblock \showarticletitle{The unreasonable effectiveness of deep features as
  a perceptual metric}. In \bibinfo{booktitle}{\emph{Proceedings of the IEEE
  Conference on Computer Vision and Pattern Recognition (CVPR)}}.
  \bibinfo{pages}{586--595}.
\newblock


\bibitem[Zhu et~al\mbox{.}(2024)]%
        {zhu2024motiongs}
\bibfield{author}{\bibinfo{person}{Ruijie Zhu}, \bibinfo{person}{Yanzhe Liang},
  \bibinfo{person}{Hanzhi Chang}, \bibinfo{person}{Jiacheng Deng},
  \bibinfo{person}{Jiahao Lu}, \bibinfo{person}{Wenfei Yang},
  \bibinfo{person}{Tianzhu Zhang}, {and} \bibinfo{person}{Yongdong Zhang}.}
  \bibinfo{year}{2024}\natexlab{}.
\newblock \showarticletitle{Motiongs: Exploring explicit motion guidance for
  deformable 3d gaussian splatting}.
\newblock \bibinfo{journal}{\emph{Advances in Neural Information Processing
  Systems}}  \bibinfo{volume}{37} (\bibinfo{year}{2024}),
  \bibinfo{pages}{101790--101817}.
\newblock


\end{thebibliography}

\appendix
\section{Appendix}

\subsection{Theorem Proof} \label{sec:proof}
\begin{theorem}[Helmholtz Decomposition Theorem]
	Let $\boldsymbol{F}: \mathbb{R}^3 \to \mathbb{R}^3$ be a sufficiently smooth vector field that decays rapidly at infinity. Then $\boldsymbol{F}$ admits a unique decomposition:
	\begin{equation}
		\boldsymbol{F}(\boldsymbol{r}) = \boldsymbol{F}_{\text{conservative}}(\boldsymbol{r}) + \boldsymbol{F}_{\text{solenoidal}}(\boldsymbol{r})
	\end{equation}
	where $\boldsymbol{F}_{\text{conservative}}$ is irrotational ($\nabla \times \boldsymbol{F}_{\text{conservative}} = \boldsymbol{0}$) and $\boldsymbol{F}_{\text{solenoidal}}$ is divergence-free ($\nabla \cdot \boldsymbol{F}_{\text{solenoidal}} = 0$).
\end{theorem}

\begin{proof}
	
	We begin with the integral identity for any vector field $\boldsymbol{F}(\boldsymbol{r})$ using the Dirac delta function:
	\begin{equation}
		\boldsymbol{F}(\boldsymbol{r}) = \int_V \boldsymbol{F}(\boldsymbol{r}') \delta(\boldsymbol{r} - \boldsymbol{r}') \, \mathrm{d}^3\boldsymbol{r}'
	\end{equation}
	where $V$ is the integration volume containing the observation point $\boldsymbol{r}$.
	
	To proceed, we employ the fundamental identity relating the Dirac delta function to the inverse Laplacian:
	\begin{equation}
		\delta(\boldsymbol{r} - \boldsymbol{r}') = -\frac{1}{4\pi} \nabla^2 \left(\frac{1}{|\boldsymbol{r} - \boldsymbol{r}'|}\right)
	\end{equation}
	
	Substituting this identity into the integral representation yields:
	\begin{equation}
		\boldsymbol{F}(\boldsymbol{r}) = -\frac{1}{4\pi} \int_V \boldsymbol{F}(\boldsymbol{r}') \nabla^2 \left(\frac{1}{|\boldsymbol{r} - \boldsymbol{r}'|}\right) \, \mathrm{d}^3\boldsymbol{r}'
	\end{equation}
	
	Since the Laplacian operator $\nabla^2$ acts on the unprimed coordinates $\boldsymbol{r}$, we can move it outside the integral:
	\begin{equation}
		\boldsymbol{F}(\boldsymbol{r}) = -\frac{1}{4\pi} \nabla^2 \int_V \frac{\boldsymbol{F}(\boldsymbol{r}')}{|\boldsymbol{r} - \boldsymbol{r}'|} \, \mathrm{d}^3\boldsymbol{r}'
	\end{equation}
	
	Now we apply the fundamental vector calculus identity to decompose the vector Laplacian:
	\begin{equation}
		\nabla^2 \boldsymbol{A} = \nabla(\nabla \cdot \boldsymbol{A}) - \nabla \times (\nabla \times \boldsymbol{A})
	\end{equation}
	
	Let $\boldsymbol{G}(\boldsymbol{r}) = \int_V \frac{\boldsymbol{F}(\boldsymbol{r}')}{|\boldsymbol{r} - \boldsymbol{r}'|} \, \mathrm{d}^3\boldsymbol{r}'$. Then:
	\begin{equation}
		\nabla^2 \boldsymbol{G} = \nabla(\nabla \cdot \boldsymbol{G}) - \nabla \times (\nabla \times \boldsymbol{G})
	\end{equation}
	
	Substituting this decomposition into our expression gives:
	\begin{align}
		\boldsymbol{F}(\boldsymbol{r}) &= -\frac{1}{4\pi} [\nabla(\nabla \cdot \boldsymbol{G}) - \nabla \times (\nabla \times \boldsymbol{G})] \\
		&= -\frac{1}{4\pi} \nabla(\nabla \cdot \boldsymbol{G}) + \frac{1}{4\pi} \nabla \times (\nabla \times \boldsymbol{G})
	\end{align}
	
	Expanding $\boldsymbol{G}$ explicitly:
	\begin{align}
		\boldsymbol{F}(\boldsymbol{r}) = &-\frac{1}{4\pi} \nabla\left(\nabla \cdot \int_V \frac{\boldsymbol{F}(\boldsymbol{r}')}{|\boldsymbol{r} - \boldsymbol{r}'|} \, \mathrm{d}^3\boldsymbol{r}'\right) \\
		&+ \frac{1}{4\pi} \nabla \times \left(\nabla \times \int_V \frac{\boldsymbol{F}(\boldsymbol{r}')}{|\boldsymbol{r} - \boldsymbol{r}'|} \, \mathrm{d}^3\boldsymbol{r}'\right)
	\end{align}
	
	This naturally leads us to identify the conservative and solenoidal components as:
	\begin{align}
		\boldsymbol{F}_c(\boldsymbol{r}) &= -\frac{1}{4\pi} \nabla\left(\nabla \cdot \int_V \frac{\boldsymbol{F}(\boldsymbol{r}')}{|\boldsymbol{r} - \boldsymbol{r}'|} \, \mathrm{d}^3\boldsymbol{r}'\right) \label{eq:conservative} \\
		\boldsymbol{F}_s(\boldsymbol{r}) &= \frac{1}{4\pi} \nabla \times \left(\nabla \times \int_V \frac{\boldsymbol{F}(\boldsymbol{r}')}{|\boldsymbol{r} - \boldsymbol{r}'|} \, \mathrm{d}^3\boldsymbol{r}'\right) \label{eq:solenoidal}
	\end{align}
	
	Thus, we have established that $\boldsymbol{F}(\boldsymbol{r}) = \boldsymbol{F}_c(\boldsymbol{r}) + \boldsymbol{F}_s(\boldsymbol{r})$.
	
	To verify the conservative nature of $\boldsymbol{F}_c$, we compute its curl:
	\begin{align}
		\nabla \times \boldsymbol{F}_c &= \nabla \times \left[-\frac{1}{4\pi} \nabla\left(\nabla \cdot \int_V \frac{\boldsymbol{F}(\boldsymbol{r}')}{|\boldsymbol{r} - \boldsymbol{r}'|} \, \mathrm{d}^3\boldsymbol{r}'\right)\right] \\
		&= -\frac{1}{4\pi} \nabla \times \nabla\left(\nabla \cdot \int_V \frac{\boldsymbol{F}(\boldsymbol{r}')}{|\boldsymbol{r} - \boldsymbol{r}'|} \, \mathrm{d}^3\boldsymbol{r}'\right)
	\end{align}
	
	By the fundamental vector calculus identity $\nabla \times \nabla \Phi = \boldsymbol{0}$ for any scalar field $\Phi$:
	\begin{equation}
		\nabla \times \boldsymbol{F}_c = \boldsymbol{0}
	\end{equation}
	
	Similarly, to verify the solenoidal property of $\boldsymbol{F}_s$, we compute its divergence:
	\begin{align}
		\nabla \cdot \boldsymbol{F}_s &= \nabla \cdot \left[\frac{1}{4\pi} \nabla \times \left(\nabla \times \int_V \frac{\boldsymbol{F}(\boldsymbol{r}')}{|\boldsymbol{r} - \boldsymbol{r}'|} \, \mathrm{d}^3\boldsymbol{r}'\right)\right] \\
		&= \frac{1}{4\pi} \nabla \cdot \nabla \times \left(\nabla \times \int_V \frac{\boldsymbol{F}(\boldsymbol{r}')}{|\boldsymbol{r} - \boldsymbol{r}'|} \, \mathrm{d}^3\boldsymbol{r}'\right)
	\end{align}
	
	By the fundamental vector calculus identity $\nabla \cdot (\nabla \times \boldsymbol{A}) = 0$ for any vector field $\boldsymbol{A}$:
	\begin{equation}
		\nabla \cdot \boldsymbol{F}_s = 0
	\end{equation}
	
	To derive the explicit integral form, we use the vector identity $\nabla \cdot (\Phi \boldsymbol{A}) = \boldsymbol{A} \cdot \nabla \Phi + \Phi \nabla \cdot \boldsymbol{A}$ and note that $\boldsymbol{F}(\boldsymbol{r}')$ depends only on primed coordinates:
	\begin{equation}
		\nabla \cdot \left(\frac{\boldsymbol{F}(\boldsymbol{r}')}{|\boldsymbol{r} - \boldsymbol{r}'|}\right) = \boldsymbol{F}(\boldsymbol{r}') \cdot \nabla \left(\frac{1}{|\boldsymbol{r} - \boldsymbol{r}'|}\right)
	\end{equation}
	
	Since $\nabla \left(\frac{1}{|\boldsymbol{r} - \boldsymbol{r}'|}\right) = -\nabla' \left(\frac{1}{|\boldsymbol{r} - \boldsymbol{r}'|}\right) = -\frac{\boldsymbol{r} - \boldsymbol{r}'}{|\boldsymbol{r} - \boldsymbol{r}'|^3}$:
	\begin{equation}
		\nabla \cdot \left(\frac{\boldsymbol{F}(\boldsymbol{r}')}{|\boldsymbol{r} - \boldsymbol{r}'|}\right) = -\boldsymbol{F}(\boldsymbol{r}') \cdot \frac{\boldsymbol{r} - \boldsymbol{r}'}{|\boldsymbol{r} - \boldsymbol{r}'|^3}
	\end{equation}
	
	Using the identity $\nabla' \cdot \left(\frac{\boldsymbol{F}(\boldsymbol{r}')}{|\boldsymbol{r} - \boldsymbol{r}'|}\right) = \frac{\nabla' \cdot \boldsymbol{F}(\boldsymbol{r}')}{|\boldsymbol{r} - \boldsymbol{r}'|} + \boldsymbol{F}(\boldsymbol{r}') \cdot \nabla' \left(\frac{1}{|\boldsymbol{r} - \boldsymbol{r}'|}\right)$ and applying the divergence theorem to eliminate surface terms (which vanish due to rapid decay at infinity), we obtain:
	\begin{align}
		\boldsymbol{F}(\boldsymbol{r}) = &\frac{1}{4\pi} \int_{\mathbb{R}^3} \frac{(\nabla' \cdot \boldsymbol{F}(\boldsymbol{r}'))(\boldsymbol{r} - \boldsymbol{r}')}{|\boldsymbol{r} - \boldsymbol{r}'|^3} \, \mathrm{d}^3\boldsymbol{r}' \\
		&+ \frac{1}{4\pi} \int_{\mathbb{R}^3} \frac{(\nabla' \times \boldsymbol{F}(\boldsymbol{r}')) \times (\boldsymbol{r} - \boldsymbol{r}')}{|\boldsymbol{r} - \boldsymbol{r}'|^3} \, \mathrm{d}^3\boldsymbol{r}'
	\end{align}
	
	To establish uniqueness, suppose there exist two decompositions:
	\begin{equation}
		\boldsymbol{F} = \boldsymbol{F}_c^{(1)} + \boldsymbol{F}_s^{(1)} = \boldsymbol{F}_c^{(2)} + \boldsymbol{F}_s^{(2)}
	\end{equation}
	
	This implies:
	\begin{equation}
		\boldsymbol{F}_c^{(1)} - \boldsymbol{F}_c^{(2)} = \boldsymbol{F}_s^{(2)} - \boldsymbol{F}_s^{(1)}
	\end{equation}
	
	Let $\boldsymbol{G} = \boldsymbol{F}_c^{(1)} - \boldsymbol{F}_c^{(2)} = \boldsymbol{F}_s^{(2)} - \boldsymbol{F}_s^{(1)}$. Since both $\boldsymbol{F}_c^{(1)}$ and $\boldsymbol{F}_c^{(2)}$ are conservative:
	\begin{equation}
		\nabla \times \boldsymbol{G} = \nabla \times (\boldsymbol{F}_c^{(1)} - \boldsymbol{F}_c^{(2)}) = \boldsymbol{0}
	\end{equation}
	
	Since both $\boldsymbol{F}_s^{(1)}$ and $\boldsymbol{F}_s^{(2)}$ are solenoidal:
	\begin{equation}
		\nabla \cdot \boldsymbol{G} = \nabla \cdot (\boldsymbol{F}_s^{(2)} - \boldsymbol{F}_s^{(1)}) = 0
	\end{equation}
	
	Therefore, $\boldsymbol{G}$ satisfies both $\nabla \times \boldsymbol{G} = \boldsymbol{0}$ and $\nabla \cdot \boldsymbol{G} = 0$. By the vector Laplacian identity:
	\begin{equation}
		\nabla^2 \boldsymbol{G} = \nabla(\nabla \cdot \boldsymbol{G}) - \nabla \times (\nabla \times \boldsymbol{G}) = \boldsymbol{0}
	\end{equation}
	
	With the boundary condition that $\boldsymbol{G}$ vanishes at infinity (due to rapid decay), the unique solution to $\nabla^2 \boldsymbol{G} = \boldsymbol{0}$ is $\boldsymbol{G} = \boldsymbol{0}$. Hence, the decomposition is unique.
\end{proof}

The Helmholtz decomposition theorem~\cite{griffiths2023introduction} establishes that any vector field can be uniquely separated into its fundamental physical components. In Hamiltonian Neural Networks, this provides the theoretical foundation for learning physically meaningful representations where the conservative component $\boldsymbol{v}_c$ preserves energy and the solenoidal component $\boldsymbol{v}_s$ preserves volume, naturally capturing the physics of deformable systems.

\subsection{Algorithm Pipeline} \label{sec:alg}
For clarity and ease of understanding, we present the overall algorithm pipeline of NeHaD in Alg. \ref{alg:nehad}.

\begin{algorithm}[t]
	\caption{NeHaD: Neural Hamiltonian Deformation Fields}
	\label{alg:nehad}
	\SetAlgoLined
	\KwIn{input images $\{\mathcal{I}_t\}_{t=1}^T$, camera parameters $\{\boldsymbol{V}_t\}_{t=1}^T$}
	\KwOut{deformed Gaussians $\mathcal{G}'$}
	
	\textbf{Initialization:}\\
	canonical Gaussians $\mathcal{G}_0 \gets \{\boldsymbol{\mu},\boldsymbol{s},\boldsymbol{r},\alpha,\boldsymbol{c}\}$;\,
	hex-plane encoder $\mathcal{E}$\;
	HNN decoder $\mathcal{D} \gets (\mathcal{M},F_1,F_2)$;\,
	adapters $\mathcal{A} \gets \{\mathcal{A}_{\mu},\mathcal{A}_{s},\mathcal{A}_{r}\}$\;
	equilibrium states $\{\boldsymbol{\mu}_{eq}^{(i)},t_{eq}^{(i)}\}$\\
	
	\For{$iter=1$ \KwTo $N_{iter}$}{
		\For{$t=1$ \KwTo $T$}{
			Sample $\boldsymbol{V}_t$\tcp*{current viewpoint}
			
			\For{each Gaussian $i$}{
				\textbf{Feature extraction}\;
				$\boldsymbol{u}_i \gets (\boldsymbol{\mu}_i,t)$\tcp*{4D coordinate}
				$\boldsymbol{f}_i \gets \prod_{k}\psi(\boldsymbol{P}_k,\pi_k(\boldsymbol{u}_i))$\tcp*{hex-plane feature}
				$\boldsymbol{h}_i \gets \mathcal{M}(\boldsymbol{f}_i)$\tcp*{latent representation}
				
				\textbf{Hamiltonian prediction}\;
				$\boldsymbol{v}_c \gets \nabla_{\boldsymbol{h}_i}F_1(\boldsymbol{h}_i)\boldsymbol{I}$\tcp*{conservative field}
				$\boldsymbol{v}_s \gets \nabla_{\boldsymbol{h}_i}F_2(\boldsymbol{h}_i)\boldsymbol{M}^{\top}$\tcp*{solenoidal field}
				$\boldsymbol{v} \gets \boldsymbol{v}_c+\boldsymbol{v}_s$\tcp*{combined field}
				$\Delta\boldsymbol{\mu}_i \gets \mathcal{A}_{\mu}(\boldsymbol{v})$;\,
				$\Delta\boldsymbol{s}_i \gets \mathcal{A}_{s}(\boldsymbol{v})$;\,
				$\Delta\boldsymbol{r}_i \gets \mathcal{A}_{r}(\boldsymbol{v})$\tcp*{offsets}
				
				\textbf{Physics-informed constraints}\;
				$\tilde{\boldsymbol{\mu}}_i \gets \boldsymbol{\mu}_i+\Delta t\cdot\Delta\boldsymbol{\mu}_i+\tfrac{(\Delta t)^2}{2}\boldsymbol{v}_c$\tcp*{symplectic position update}
				$\phi_i \gets 2\,\mathrm{atan2}(\|\Delta \boldsymbol{g}_i\|,\Delta w_i)$\;
				$\phi_i' \gets \phi_{max}\tanh(\phi_i/\phi_{max})$\tcp*{rigidity clamp}
				$\Delta \boldsymbol{r}_i' \gets [\cos(\phi_i'/2),\,\sin(\phi_i'/2)\tfrac{\Delta \boldsymbol{g}_i}{\|\Delta \boldsymbol{g}_i\|}]^{\top}$\;
				$\boldsymbol{r}_i' \gets \mathcal{N}(\boldsymbol{r}_i \otimes \Delta\boldsymbol{r}_i')$\tcp*{unit quaternion update}
				
				\textbf{Boltzmann equilibrium masks}\;
				$\Delta d_i \gets \|\boldsymbol{\mu}_i-\boldsymbol{\mu}_{eq}^{(i)}\|_2/\sigma_s$;\,
				$\Delta\tau_i \gets (t-t_{eq}^{(i)})/\sigma_t$\tcp*{spatial and temporal deviations}
				$E_{st}^{(i)} \gets \tfrac{1}{2}(\Delta d_i^2+\Delta\tau_i^2)+\lambda\Delta d_i\Delta\tau_i$\;
				$M_{pos}^{(i)} \gets (1-\gamma)\exp(-\beta E_{st}^{(i)})+\gamma$\tcp*{position mask}
				$E_t^{(i)} \gets \tfrac{1}{2}((t-t_{eq}^{(i)})/\sigma_t)^2$\;
				$M_{scale}^{(i)} \gets (1-\gamma)\exp(-\beta E_t^{(i)})+\gamma$\tcp*{scaling mask}
				
				\textbf{Masked updates}\;
				$\boldsymbol{\mu}_i' \gets \tilde{\boldsymbol{\mu}}_i\odot(1-M_{pos}^{(i)})+\boldsymbol{\mu}_i\odot M_{pos}^{(i)}$\tcp*{final position}
				$\boldsymbol{s}_i' \gets \boldsymbol{s}_i+\Delta\boldsymbol{s}_i\odot(1-M_{scale}^{(i)})$\tcp*{final scaling}
			}
			
			\textbf{Rendering and optimization}\;
			$\mathcal{G}' \gets \{\boldsymbol{\mu}',\boldsymbol{s}',\boldsymbol{r}',\alpha,\boldsymbol{c}\}$\;
			$\boldsymbol{X} \gets \mathcal{R}(\boldsymbol{V}_t,\mathcal{G}')$\tcp*{differentiable rasterization}
			$\mathcal{L} \gets (1-\lambda)\mathcal{L}_1+\lambda\mathcal{L}_{DSSIM}+\mathcal{L}_{TV}$\tcp*{loss}
			Update all parameters by backpropagation\;
		}
	}
	
	\Return{Optimized deformation model}
\end{algorithm}

\subsection{Implementation Details} \label{sec:detail}
\noindent \textbf{Datasets.} We evaluate on both monocular and multi-view dynamic scene datasets, including synthetic and real-world scenes:

\emph{Synthetic Dataset.} We use D-NeRF~\cite{dnerf} for synthetic scene evaluation. D-NeRF contains 8 scenes with large-scale movements and non-Lambertian materials, posing challenges for dynamic modeling. We render at 800$\times$800 resolution.

\emph{Real-World Datasets.} We evaluate on HyperNeRF~\cite{hypernerf} and DyNeRF~\cite{dynerf} for real-world scenes. HyperNeRF captures scenes using 1 or 2 cameras with straightforward motions, while DyNeRF contains 6 ten-second videos recorded at 30 FPS using 15 to 20 static cameras. Unlike synthetic data, camera poses are estimated using COLMAP~\cite{colmap}. We report quantitative results on HyperNeRF's ``vrig'' (validation rig) scenes at 536$\times$960 resolution and DyNeRF at 1352$\times$1014 resolution.

\noindent \textbf{Baselines.} We compare NeHaD with several state-of-the-art methods~\cite{4dgs, sc-gs, k-planes, hexplane, grid4d, swift4d, saro-gs, spacetime-gaussians, tineuvox, deformable3d}. TiNeuVox, K-Planes, and HexPlane are NeRF-based approaches: TiNeuVox uses 3D grids while K-Planes and HexPlane employ plane-based explicit representations. 4DGS, DeformGS, SC-GS, STG, Grid4D, SaRO-GS, and Swift4D are Gaussian-based methods. 4DGS, DeformGS and SaRO-GS combine plane-based representations with MLP-based deformation fields. SC-GS builds on DeformGS using sparse control points for dynamic rendering and editing. STG extends 3DGS to 4D spacetime with temporal opacity and polynomial motion for real-time synthesis. Grid4D and Swift4D both use hash encoding as explicit representations to decompose spatial-temporal Gaussian deformations.

\noindent \textbf{Evaluation Metrics.} We employ multiple metrics to evaluate our method. For rendering quality, we use PSNR, SSIM~\cite{ssim}, DSSIM, MS-SSIM, and LPIPS~\cite{lpips}. For rendering efficiency, we measure FPS. All results are averaged across all scenes unless specified otherwise. Due to the lack of standard metrics for physical plausibility, we evaluate this aspect through qualitative comparisons and dynamics visualization.

\noindent \textbf{Hyperparameters.} Our hyperparameter settings largely follow 4DGS~\cite{4dgs}. For all datasets, a basic resolution of 64 is used for the hex-plane encoder with upsampling factors of 2 and 4. The learning rate starts at 0.0016 and is decayed to 0.00016 by the end of training. The pruning interval is changed to 8000 and only a single upsampling rate of the hex-plane encoder as 2 because the relatively simple structures in D-NeRF~\cite{dnerf} dataset. The Gaussian deformation decoder consists of a single HNN trained with Adam optimizer~\cite{adam} at a learning rate of 0.001. Most scenes are trained for 20,000 iterations with a batch size of 1, while complex real-world scenes such as \emph{flame\_salmon} require additional training, and we empirically set to 30,000 iterations. All experiments are conducted on a single NVIDIA RTX 3090 GPU.

\subsection{Detailed Methodology of Streaming} \label{sec:stream}
We adapt NeHaD to streaming through two enhancements: (a) scale-aware anisotropic MipMapping for efficient multi-level texture sampling, and (b) layered progressive optimization for global level-of-detail (LOD) rendering, as illustrated in Fig. \ref{fig:streaming}.
\begin{figure}[t]
	\centering
	\includegraphics[width=1\columnwidth]{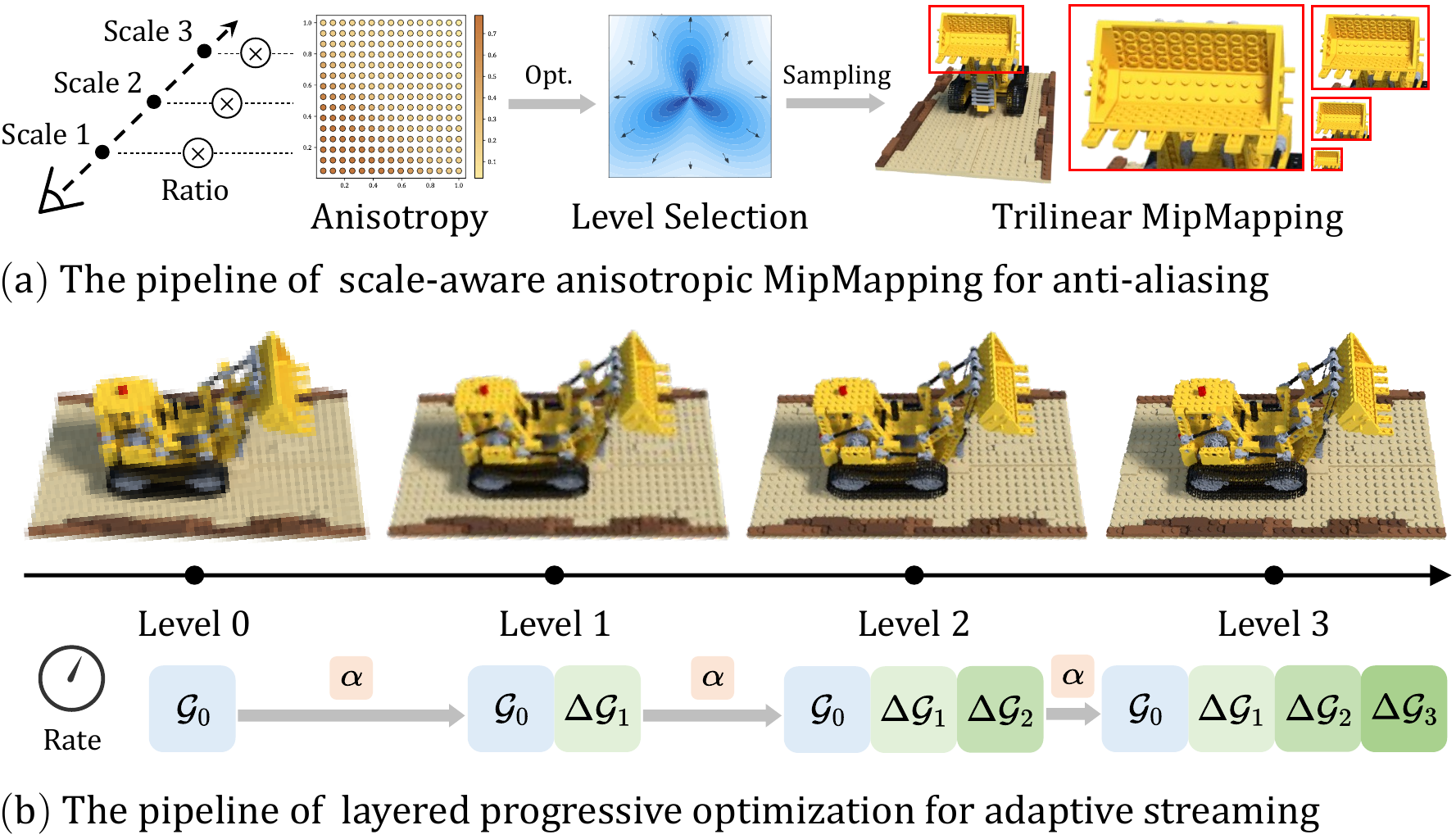}
	\caption{\textbf{Methodology overview of adapting NeHaD to streaming.} (a) Scale-aware anisotropic mipmapping pipeline for anti-aliasing: different scales from varying viewing distances are analyzed using anisotropic weights, followed by mipmap level selection and trilinear MipMapping for texture sampling. (b) Layered progressive optimization pipeline for adaptive streaming: training proceeds progressively across different LOD levels (from lower to higher resolution), with Gaussian residuals obtained adaptively based on opacity to enhance the ground Gaussian (level 0) in a layered manner for increased detail at higher rendering rates.}
	\label{fig:streaming}
\end{figure}

To formulate the streaming strategy mathematically, given a scale factor $\boldsymbol{s} \in \mathbb{R}^3$ representing sampling scales in spatial dimensions, we aim to compute an adaptive mipmap level vector $\boldsymbol{l} \in \mathbb{R}^4$ for anisotropic filtering requirements. The input scale is normalized by $\text{clamp}(\cdot, \boldsymbol{a}, \boldsymbol{b})$, constraining $\boldsymbol{s}$ to $[\frac{\boldsymbol{s}}{2}, \frac{\boldsymbol{s}}{2} \odot \boldsymbol{r}]$, where $\boldsymbol{r} \in \mathbb{N}^3$ is the resolution ratio vector and $\odot$ represents element-wise multiplication. The level $\boldsymbol{l}$ is computed using normalized $\tilde{\boldsymbol{s}}$ and original scale $\boldsymbol{s}$ through binary logarithmic operations, \ie $\boldsymbol{l} = \log_2(\frac{2\tilde{\boldsymbol{s}}}{\boldsymbol{s}})$.

The dominant axis is determined by the ratio of current level to maximum level. Given the principle level $\boldsymbol{L}$ from this dominant axis, the final level $\hat{\boldsymbol{l}}$ is computed as:
\begin{equation}
	\hat{\boldsymbol{l}}=\boldsymbol{L}-\beta(\rho)\cdot\left(\boldsymbol{L}-\bar{L}\boldsymbol{1}\right), \quad \beta(\rho)=\frac{\mathrm{tanh}(\rho/3-1)}{1+\mathrm{tanh}(\rho/3-1)},
\end{equation}
where $\bar{L}$ is the mean principle level across spatial dimensions, $\beta(\rho)$ is the anisotropy regularization factor smoothly interpolating between isotropic ($\beta\to0$) and anisotropic ($\beta\to1$) regimes, and $\rho$ is computed by dividing the maximum and minimum $\boldsymbol{l}$. Using the selected final level $\hat{\boldsymbol{l}}$, we perform trilinear MipMapping through the nvdiffrast~\cite{laine2020modular} library.

For layered progressive optimization, following \cite{lapisgs}, we begin with the lowest quality level (lowest resolution images) and initially optimize ground Gaussian splats $\mathcal{G}_0$. As training progresses, higher resolution views are considered at each layer. Gaussian splats at various quality levels are represented as $\{\mathcal{G}_i\}_{i=0}^N$ (with total $N$ levels). The layered progressive optimization process is represented as $\{\mathcal{G}_0, \{\Delta\mathcal{G}_i\}_{i=1}^N\}$:
\begin{equation}
	\mathcal{G}_i = \mathcal{G}_0 + \sum_{j=1}^{i} \Delta \mathcal{G}_j, \quad j\in\{1,2,\dots,N\},
\end{equation}
where $\Delta \mathcal{G}_j$ is the $j$-th enhancement layer. To achieve smooth transitions between low and high quality layers, we follow \cite{lapisgs, lightgaussian, sun2024multi} by adjusting the opacity threshold, allowing the system to fine-tune transmitted rates and adapt to varying network bandwidths or device capabilities.

\subsection{Notation Table}
The notations used throughout the paper can be found in Tab. \ref{tbl:notation}.

\renewcommand{\arraystretch}{1.05}
\begin{table}[t]
	\caption{Summary of notations used throughout the paper}
	\label{tbl:notation}
	\centering
	\small
	\rowcolors{3}{gray!7}{white}
	\begin{tabular}{p{0.22\columnwidth} p{0.7\columnwidth}}
		\toprule
		\textbf{Symbol} & \textbf{Description} \\
		\midrule
		
		$\mathcal{G}$ & \desc{Gaussian primitive} \\
		$\boldsymbol{\mu} \in \mathbb{R}^3$ & \desc{Position} \\
		$\boldsymbol{s} \in \mathbb{R}^3$ & \desc{Scaling} \\
		$\boldsymbol{r} \in \mathbb{R}^4$ & \desc{Rotation quaternion} \\
		$\alpha \in \mathbb{R}$ & \desc{Opacity (scalar)} \\
		$\boldsymbol{c} \in \mathbb{R}^n$ & \desc{Color (spherical harmonics, SH)} \\
		$\boldsymbol{\Sigma} \in \mathbb{R}^{3\times3}$ & \desc{Covariance matrix} \\
		
		$\Delta \mathcal{G}$ & \desc{Gaussian deformation} \\
		$\Delta \boldsymbol{\mu}, \Delta \boldsymbol{s}, \Delta \boldsymbol{r}$ & \desc{Position, scaling, and rotation offsets} \\
		$\mathcal{E}$ & \desc{Hex-plane encoder} \\
		$\mathcal{D}$ & \desc{Deformation decoder} \\
		$\boldsymbol{P}_k$ & \desc{Plane $k$ in hex-plane factorization} \\
		$\boldsymbol{f}$ & \desc{Hex-plane spatial–temporal features} \\
		$\boldsymbol{u}=(x,y,z,t)$ & \desc{4D coordinate} \\
		
		$\boldsymbol{q}, \boldsymbol{p} \in \mathbb{R}^d$ & \desc{Coordinates (position, momentum)} \\
		$\mathcal{H}$ & \desc{Hamiltonian (total energy)} \\
		$\mathcal{U}$ & \desc{Potential energy} \\
		$\mathcal{K}$ & \desc{Kinetic energy} \\
		$\boldsymbol{S}_{\mathcal{H}}$ & \desc{Symplectic gradient} \\
		$\boldsymbol{M}$ & \desc{Permutation tensor} \\
		
		$\mathcal{M}$ & \desc{MLP baseline in HNN} \\
		$\boldsymbol{h} \in \mathbb{R}^W$ & \desc{Latent representation} \\
		$F_1, F_2$ & \desc{Scalar functions for vector fields} \\
		$\boldsymbol{v}_c, \boldsymbol{v}_s$ & \desc{Conservative and solenoidal fields} \\
		$\mathcal{A}_{\mu}, \mathcal{A}_{s}, \mathcal{A}_{r}$ & \desc{Gaussian attribute adapters} \\
		$\boldsymbol{\theta}$ & \desc{Network parameters} \\
		
		$\boldsymbol{\mu}_{eq}^{(i)}, t_{eq}^{(i)}$ & \desc{Spatial/temporal equilibrium states} \\
		$E_{st}^{(i)}, E_{t}^{(i)}$ & \desc{Spatial-temporal/temporal energy deviation} \\
		$M_{pos}^{(i)}, M_{scale}^{(i)}$ & \desc{Equilibrium masks for position/scaling} \\
		$\beta=1/T$ & \desc{Inverse temperature} \\
		$\gamma$ & \desc{Minimum dynamic responsiveness} \\
		$\sigma_s, \sigma_t$ & \desc{Spatial and temporal sensitivity scales} \\
		$\lambda$ & \desc{Coupling coefficient} \\
		
		$\boldsymbol{F}_i$ & \desc{Force field} \\
		$\tilde{\boldsymbol{\mu}}_i$ & \desc{Symplectically integrated position} \\
		$\phi_i$ & \desc{Rotation angle} \\
		$\phi_{max}$ & \desc{Maximum allowable rotation} \\
		$\Delta \boldsymbol{g}_i, \Delta w_i$ & \desc{Vector/scalar parts of quaternion increment} \\
		
		$\mathcal{R}$ & \desc{Differentiable rasterization operator} \\
		$\boldsymbol{X}$ & \desc{Rendered image (novel view)} \\
		$\mathcal{I}_t$ & \desc{Input image at timestamp $t$} \\
		$\boldsymbol{V}_t$ & \desc{Camera view matrix at timestamp $t$} \\
		$\mathcal{L}, \mathcal{L}_1, \dots, \mathcal{L}_{TV}$ & \desc{Loss terms} \\
		$\Delta t,\; T,\; N_{iter}$ & \desc{Time step, total timestamps, and training iterations} \\
		\bottomrule
	\end{tabular}
\end{table}

\end{document}